\documentclass[aps, pra, reprint, amsmath, amssymb]{revtex4-2}
\usepackage{graphicx}
\usepackage{amsthm}
\usepackage[colorlinks=true, citecolor=blue, urlcolor=blue]{hyperref}
\usepackage{booktabs}
\newcommand{\Tr}{\mathrm{Tr}}
\newcommand{\ket}[1]{| #1 \rangle}

\newcommand{\rank}{\operatorname{rank}}
\theoremstyle{definition}
\newtheorem{lemma}{Lemma}
\newtheorem{theorem}{Theorem}
\newtheorem{proposition}{Proposition}

\begin{document}

\title{Quantifying nonstabilizerness of quantum codes from difference structure and cubic phases}
\author{Yuan Liu}
\author{Ke-Mi Xu}
\email{xukemi@bit.edu.cn}
\affiliation{MIIT Key Laboratory of Complex-field Intelligent Exploration, School of Optics and Photonics, Beijing Institute of Technology, Beijing 100081, China}

\begin{abstract}
Nonstabilizerness (magic) is the quantum resource that, together with stabilizer operations, makes universal quantum computation possible. However, quantifying nonstabilizerness of quantum error correction codes is difficult. The standard measure sums over exponentially many Pauli operators, and for most codes no efficient quantitative theory has been available. In this work we remove this obstacle via a structural observation: the Clifford sector carries no nonstabilizerness and can be removed by a Clifford transformation, leaving an order-three object whose nonstabilizerness is an exactly evaluable fourth moment. Removing this inert background yields closed forms for several code families, namely a formula for all Dicke states, reducing a $n$-qubit case from $4^n$ terms to a short binomial sum, a bound for cubic-phase codes (twisted quantum doubles and non-Abelian topological order), saturated only by the $D_4$ code, and a cyclic/zero criterion for group multiplication states. For codeword-stabilized (CWS) codes, nonstabilizerness is determined by the codeword difference structure and remains unchanged under coset closure, with Sidon sets giving extremal values and nonstabilizerness quantitatively limiting diagonal transversal non-Clifford gates. Together these results reduce the computation of nonstabilizerness for a broad class of codes to finite classical counting problems,  and identify the maximally magical codes among them.
\end{abstract}

\maketitle

\section{Introduction}

Nonstabilizerness, also termed magic, quantifies the departure of a quantum state from the stabilizer polytope. Together with stabilizer operations, it is the resource that enables universal quantum computation \cite{Bravyi05PRA, Howard17Mar, Chitambar19RMP, Liu22PRXQ}. A convenient measure of nonstabilizerness is the second-order stabilizer R\'enyi entropy (SRE) \cite{Leone22Feb, Haug23PRXQ}
\begin{equation}
\label{eq:M2}
  M_2(\ket{\psi})=-\log_2\Big(2^{-n}\sum_{P\in\mathcal P_n}\langle\psi|P|\psi\rangle^4\Big),
\end{equation}
where $\mathcal{P}_n:=\{I, X, Y, Z\}^{\otimes n}$ is the $n$-qubit Pauli group. The global phases of $P$ are suppressed, since they do not enter the fourth moment. For a generic state the sum runs over all $4^n$ Pauli operators, so a direct evaluation is impractical already for a few tens of qubits \cite{Xiao26Aug, Sierant26Quantum, Varela26PRXQ}.

Quantum error correction codes store logical information in their code states, which are protected by fault-tolerant devices. Quantifying the nonstabilizerness of quantum codes (i.e., nonstabilizerness of the code states) is thus of great importance. However, prior works provide exact values only in special cases, and for whole code families only qualitative facts. Closed forms are known for the $W$ state \cite{Odavic23SciPost, Catalano24arXiv}, for fixed-excitation Dicke \cite{Liu25TIMPS} and hypergraph states \cite{Chen24Quantum, Kagamihara26arXiv}, and the permutation-invariant machinery of \cite{Passarelli24PRA} evaluates symmetric states efficiently in a Dicke basis, but these do not extend to full families or to large codes. Twisted quantum doubles furnish non-Pauli topological stabilizer codes \cite{Fuente21Quantum, Ellison22PRXQ}, whose SRE encodes fusion rules \cite{Hoshino26Feb} and whose non-Abelian orders carry extensive, long-ranged nonstabilizerness \cite{Kobayashi26Jun, Zhang26arXiv, Korbany25Oct, Wei26arXiv}, but no exact values have been reported. Exact nonstabilizerness of most quantum codes has therefore remained out of reach.

Efficient quantification requires simplifying the sum in Eq.~\eqref{eq:M2}. A Clifford unitary conjugates the Pauli group into itself and leaves $M_2$ invariant \cite{Leone22Feb}, so any part of a state generated by Clifford operations is invisible to $M_2$ and removable. For a wide range of states this inert part is exactly the quadratic (order-two), Clifford-implementable sector, and removing it isolates an order-three (cubic) object that carries all the nonstabilizerness. In the language of the Pauli and Clifford hierarchy \cite{Bravyi13Apr}, the Pauli group is a central extension of $\mathbb{F}_2^{2n}$ classified by the symplectic form, the Clifford layer is the quadratic sector, and nonstabilizerness is the order-three obstruction to quadraticity. Removing the Clifford background reduces the nonstabilizerness of each code family to a finite, exactly evaluable quantity. The reduction is carried out by two fourth-moment dictionaries, both instances of the same Parseval identity. On the indicator side a code state reduces to a subset state whose nonstabilizerness is the additive energy of its support, and on the phase side to a phase state whose nonstabilizerness is the rank data of a cubic Hessian.

Applied to four code families, the dictionaries yield several results. For permutation-invariant codes we obtain a closed form for all Dicke states, valid for every excitation number $w$ and block length $n$, so that the nonstabilizerness of a $100$-qubit Dicke state reduces from a sum over $4^{100}$ terms to a short binomial sum. For  codeword-stabilized (CWS) codes the code state reduces to a subset state supported on the classical label code, so the nonstabilizerness is determined by the codeword difference structure and is maximized by Sidon sets: $M_2\le 3k-\log_2(7\cdot2^k-6)$ for $k$ logical qubits, with equality for the $W$ state, whose support is a Sidon set. The value is invariant under coset closure, which yields non-stabilizer codes of constant nonstabilizerness and arbitrarily many logical qubits, and it is exactly computable for the Kerdock codes \cite{Kerdock1972182, Hammons94TIT}, $M_2=3(\eta-1)-\log_2(7\cdot2^{\eta-1}-6)$ at code length $n=2^\eta$ for even $\eta$, and $M_2=9-\log_2 50\approx3.356$ for the Nordstrom--Robinson code. The same counting bounds the transversal gates of the code: a diagonal transversal gate that is non-Clifford on $t$ coordinates forces $M_2\le 2(k-t)$. For cubic-phase codes, the decoupled form of twisted quantum doubles and non-Abelian topological order, we obtain the exact nonstabilizerness together with a lower bound saturated only by the determinant on $\mathbb{F}_2^3$ (the $D_4$ code). For group multiplication states the nonstabilizerness vanishes if and only if the group law is cyclic, and the minimal non-cyclic case carries exactly the nonstabilizerness of the controlled-controlled-Z (CCZ) state. Together these results reduce the computation of nonstabilizerness for a broad class of codes to finite classical counting problems, and identify the maximally magical codes among them. Our work thus provides a general and efficient methodology to evaluate nonstabilizerness of quantum codes.

The rest of the paper is organized as follows. Section~\ref{sec:principle} develops the general principle. We first show how the Clifford background of a code state is separated off and removed by a Clifford transformation, and then derive the two fourth-moment dictionaries. Section~\ref{Dicke} applies the indicator dictionary to permutation-invariant codes and obtains a single closed form for all Dicke states. Section~\ref{sec:CWS} treats CWS codes. We prove a graph--code decoupling lemma that removes the graph state, so that the nonstabilizerness is carried by the classical code alone, and then determine the extremal codes and the maximal nonstabilizerness of a code space, the effect of coset closure, exact values for the Kerdock codes, and the bound that the same difference structure places on diagonal transversal gates. Section~\ref{sec:other} applies the phase dictionary to cubic-phase codes, i.e., twisted quantum doubles and non-Abelian topological order, and to the multiplication states of quantum doubles. We conclude in Section~\ref{sec:conc}.

\section{General Principle}
\label{sec:principle}

\subsection{Decoupling the Quadratic Background}

The decoupling rests on a degree filtration of the Pauli and Clifford hierarchy. We work on $n$ qubits and write $\mathbb{F}_2^n$ for the set of $n$-bit strings $x=(x_1,\dots,x_n)$ with $x_i\in\{0,1\}$. This is also the label space of the Pauli group. A Pauli operator is written as (up to a phase) $X^{u}Z^{z}=\prod_{i=1}^n X_i^{u_i}Z_i^{z_i}$, $u,z\in\mathbb{F}_2^n$, with $X\ket b=\ket{b\oplus1}$ and $Z\ket b=(-1)^b\ket b$. The commutation rule $X^{u}Z^{z}\,X^{u'}Z^{z'}=(-1)^{u\cdot z'+u'\cdot z}\,X^{u'}Z^{z'}\,X^{u}Z^{z}$ exhibits $\mathcal P_n$ as a central extension of $\mathbb{F}_2^{2n}$ by $\mathbb{Z}_2$, with extension class the symplectic form $\omega((u,z),(u',z'))=u\cdot z'+u'\cdot z$. This is the Heisenberg group. A Clifford unitary is one that conjugates $\mathcal P_n$ into itself; on labels it acts by a symplectic transformation $\mathrm{Sp}(2n,\mathbb{F}_2)$, up to a phase that is itself quadratic in the labels. In this sense ``Clifford'' and ``quadratic'' are synonymous, i.e., they define the order-two sector.

The order-one and order-two content of a state is read off from its phases. A stabilizer state in the computational basis has the form
\begin{equation}
  \ket{\psi}\propto\sum_{x}(-1)^{q(x)}\ket x,
\end{equation}
with $q$ a quadratic Boolean function, $q(x)=\sum_i a_i x_i+\sum_{i<j}b_{ij}x_ix_j$. The canonical example is the graph state of a graph $G$ (with the set of edges denoted by $E(G)$) \cite{Raussendorf01May, Briegel01Jan, Hein04PRA}, whose phase is $q_G(x)=\sum_{\{i,j\}\in E(G)}x_ix_j$ \cite{Dehaene03PRA}.  Explicitly, the diagonal unitary $\prod_iZ_i^{a_i}\prod_{i<j}\mathrm{CZ}_{ij}^{b_{ij}}$ that implements $(-1)^{q(x)}$ is Clifford, and the nonstabilizerness of $\ket{\psi}\propto\sum_x(-1)^{f(x)}\ket x$ is determined solely by the part of $f$ of degree at least three.

To isolate the high-degree part we introduce the discrete (Boolean) derivatives. For $f:\mathbb{F}_2^n\to\mathbb{F}_2$ the discrete derivative in direction $u$ is defined by $\Delta_u f(x)=f(x)+f(x+u)$. It is nilpotent ($\Delta_u^2=0$), whereas derivatives in distinct directions commute but do not cancel ( $\Delta_u\Delta_v\neq0$ for $u\neq v$). Iteration gives the generally nonzero second and third derivatives
\begin{equation}
\begin{aligned}
  \Delta_u\Delta_v f(x)&=f(x)+f(x{+}u)+f(x{+}v)+f(x{+}u{+}v),\\
  \Delta_u\Delta_v\Delta_w f&=\sum_{S\subseteq\{u,v,w\}}f\Big(\sum_{s\in S}s\Big),
\end{aligned}
\end{equation}
the former being the ``parallelogram curvature'' of $f$. Nilpotency acts only on repeated directions, so $\Delta_u\Delta_v\Delta_w f$ is an alternating trilinear form in $u,v,w$, vanishing whenever two of them coincide (hence it is a three-cocycle in the language of group cohomology). The order-two sector is invisible to $M_2$, and what survives is the three-cocycle $\Delta_u\Delta_v\Delta_w f$.

We therefore ask: after a Clifford transformation that removes the quadratic background, what does the code state become? For all the codes we consider, the answer is one of two canonical forms (Fig.~\ref{Fig1}). On the indicator side the state is a subset state
\begin{equation}
  \ket S=|S|^{-1/2}\sum_{x\in S}\ket x,
\end{equation}
whose nonstabilizerness is governed by the nonlinearity of the support $S$. On the phase side the state is a phase state
\begin{equation}
  \ket{C_f}=2^{-n/2}\sum_{x\in\mathbb{F}_2^n}(-1)^{f(x)}\ket x,
\end{equation}
whose nonstabilizerness is governed by the cubic part of $f$. Each canonical form admits an fourth-moment dictionary. We derive the two dictionaries in the next subsection. The applications then require identifying, for each code family, which canonical form its code state reduces to.

\begin{figure}
  \centering
  \includegraphics[width=8cm]{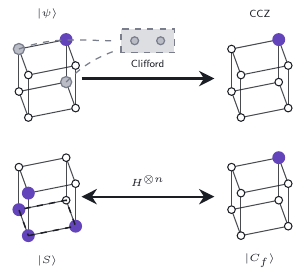}
  \caption{Decoupling the quadratic background. Top: The phase $(-1)^{f(x)}$ of a code state $\ket\psi=\sum_{x\in\mathbb F_2^3}(-1)^{f(x)}\ket x$ separates into a quadratic part (gray vertices) and a cubic core (violet). A Clifford unitary removes the quadratic background, leaving the cubic core. For $n=3$ the unique cubic is the determinant $f=x_1x_2x_3$, the CCZ state. Bottom: The surviving cubic object appears in two complementary forms: the subset state $\ket S$ (left; the Sidon set $S=\{000,100,010,001\}$, whose broken parallelogram marks the absence of additive structure) and the phase state $\ket{C_f}$ (right, the CCZ state).}
\label{Fig1}
\end{figure}

\subsection{Two Fourth-moment Dictionaries}

The derivation depends on the fourth-moment form of Parseval's identity. For any real-valued function $a:\mathbb{F}_2^n\to\mathbb{R}$, we have
\begin{equation}
\label{eq:parseval}
  \sum_{z\in\mathbb{F}_2^n}\Big|\sum_{x\in\mathbb{F}_2^n}a(x)(-1)^{z\cdot x}\Big|^4=2^n\sum_{v\in\mathbb{F}_2^n}\Big|\sum_{x\in\mathbb{F}_2^n}a(x)a(x\oplus v)\Big|^2 .
\end{equation}

For a set $S\subset\mathbb{F}_2^n$ with state $\ket S=|S|^{-1/2}\sum_{x\in S}\ket x$ and a Pauli operator $P=X^uZ^z$, the expectation value is
\begin{equation}
 c(x,z):=\langle S|X^xZ^z|S\rangle=\frac{1}{m}\sum_{w\in S\cap(S\oplus x)}(-1)^{z\cdot w},
\label{eq:corr}
\end{equation}
because $X^xZ^z|w\rangle\propto|w\oplus x\rangle$ picks out the codewords paired by the shift $x$.
Fixing $u$ and putting 
\begin{equation}\label{eq:A}
  S_u:=S\cap(S\oplus u), \quad A(u):=|S_u|,
\end{equation}
applying Eq.~\eqref{eq:parseval} to $a=\mathbf1_{S_u}$ gives $\sum_z|\langle S|X^uZ^z|S\rangle|^4=|S|^{-4}2^n E(S_u)$, where
\begin{equation}
  E(S):=\sum_{y\in\mathbb{F}_2^n}|S\cap(S\oplus y)|^2
\end{equation}
is the \emph{additive energy} of $S$ \cite{Tao_Vu_2006}. Summing over $u$ and inserting into Eq.~\eqref{eq:M2} yields the indicator-side dictionary,
\begin{equation}
\label{eq:indicator}
\begin{aligned}
  M_2(\ket S)&=4\log_2|S|-\log_2 E^{(1)}(S),\\
  E^{(1)}(S)&=\sum_{x\in\mathbb{F}_2^n}E(S\cap(S\oplus x)).
\end{aligned}
\end{equation}
$E^{(1)}$ is termed the \emph{parallelogram energy}.
To get an intuitive understanding of the naming, we observe that the additive energy ($\#\,\cdot$ denotes the number of elements of a set, which is used interchangeably with $|\cdot|$)
\begin{equation}
 E(S)=\sum_{x}A(x)^2=\#\{(a,b,c,d)\in S^4 \mid a\oplus b\oplus c\oplus d=0\},
\label{eq:E}
\end{equation}
counts additive quadruples, and the parallelogram energy
\begin{equation}
\begin{aligned}
 E^{(1)}(S)&=\sum_{x}E\bigl(S\cap(S\oplus x)\bigr)=\sum_{x,y}A_2(x,y)^2,\\
 A_2(x,y):&=\#\{w \mid w, w\oplus x, w\oplus y, w\oplus x\oplus y\in S\},
\end{aligned}
\label{eq:E1}
\end{equation}
counts the affine $2$-planes (parallelograms) contained in $S$.

A linear code $S$ has $E^{(1)}=|S|^4$ and hence $M_2=0$, as it must; nonstabilizerness arises only from the nonlinearity of $S$. This dictionary is the tool behind the analysis of permutation-invariant codes and CWS codes, and it applies verbatim to any code whose code state is Clifford-equivalent to a subset state.

The phase side is the Fourier-dual counterpart of the indicator side. Instead of a sparse $\{0,1\}$ indicator on a subset $S$, the relevant function takes both signs, $a=\pm1$, on all of $\mathbb{F}_2^n$. For $P=X^uZ^z$ we have
\begin{equation}
\begin{aligned}
  \langle C_f|X^uZ^z|C_f\rangle&=2^{-n}\sum_x(-1)^{f(x)+f(x+u)+z\cdot x}\\
  &=2^{-n}\sum_x a_u(x)(-1)^{z\cdot x},
\end{aligned}
\end{equation}
with $a_u(x)=(-1)^{\Delta_uf(x)}$ real. Applying Eq.~\eqref{eq:parseval} and $a_u(x)a_u(x+v)=(-1)^{\Delta_u\Delta_v f(x)}$ gives
\begin{equation}
\label{eq:phase_dict}
  M_2(\ket{C_f})=n-\log_2\Big(2^{-3n}\sum_{u,v\in\mathbb{F}_2^n}\Big|\sum_x(-1)^{\Delta_u\Delta_v f(x)}\Big|^2\Big),
\end{equation}
a fourth moment of the second derivative of $f$. 

For a cubic $f$, $\Delta_u\Delta_v\Delta_w f$ does not depend on the base point $x$, and $\Delta_u\Delta_v f(x)$ is affine in $x$. An affine $g$ satisfies $g(x)=g(0)\oplus\Delta_x g(0)$, so for $g=\Delta_u\Delta_v f$ we have $\Delta_u\Delta_v f(x)=\Delta_u\Delta_v\Delta_x f\oplus c(u,v)$, $c(u,v):=\Delta_u\Delta_v f(0)$. The $x$-dependent term is the three-cocycle $B(u,v,x):=\Delta_u\Delta_v\Delta_x f$, and the constant $c(u,v)$ is irrelevant as it only flips the overall sign of the Fourier coefficient. The sum over $x$ then selects $B(u,v,\cdot)\equiv0$, and
\begin{equation}
\label{eq:cubic_dict}
\begin{aligned}
  M_2(\ket{C_f})&=2n-\log_2 N(B),\\
  N(B)&=\#\big\{(u,v)\in\mathbb{F}_2^{2n} \mid B(u,v,\cdot)\equiv0\big\}.
\end{aligned}
\end{equation}
Eq.~\eqref{eq:cubic_dict} is the degree-three instance of the hypergraph-state SRE formula of Ref.~\cite{Chen24Quantum}, in the rank form given in Ref.~\cite{Kagamihara26arXiv}; within the two-dictionary framework it is recovered from the same Parseval identity as the indicator side, and expressed through the trilinear form $B$. The nonstabilizerness of a cubic code is therefore determined entirely by the trilinear form $B$. Geometrically, $N(B)$ counts singular directions of the Hessian. For each $u$ the slice $M_u(v,w)=B(u,v,w)$ is a bilinear form on $\mathbb{F}_2^n$, and $B(u,v,\cdot)\equiv0$ means $v\in\ker M_u$, so $N(B)=\sum_u2^{n-\rank M_u}$.

Each slice is \emph{alternating}, i.e., $M_u(v,v)=B(u,v,v)=B(v,v,u)$, and $B(v,v,\cdot)=\Delta_v^2f(\cdot)$ vanishes identically on $\mathbb{F}_2$. Consequently $M_u$ is alternating and $\rank M_u$ is even for every $u$. The direction $u=0$ contributes $2^n$; every nonzero direction contributes at least $1$ (even $n$) or $2$ (odd $n$, where the maximal even rank is $n-1$). Hence
\begin{equation}
  N(B)\ge\begin{cases}2^{n+1}-1,& n\ \text{even},\\ 3\cdot2^n-2,& n\ \text{odd}.\end{cases}
\end{equation}
The bound is attainable only if every $M_u$ with $u\neq0$ is nonsingular. For $n=3$ the only cubic form (up to scaling) is the determinant $f=x_1x_2x_3$, for which $M_u(v,w)=\det(u,v,w)$ is a nonvanishing alternating form on $\mathbb{F}_2^3$, of rank $2=n-1$ whenever $u\neq0$. Thus $N(\det)=2^3+7\cdot2^{3-2}=22$, and the bound is saturated:
\begin{equation}
  M_2(\mathrm{CCZ})=6-\log_2 22,
\end{equation}
which coincides with the value obtained for the three-qubit CCZ (hypergraph) state in Ref.~\cite{Chen24Quantum}. The maximally magical cubic form is thus unique and three-dimensional. On the phase side it is the CCZ state, while on the indicator side the corresponding Sidon set \cite{Sidon1932, Czerwinski24AdvMathCommun} $S=\{000,100,010,001\}$ realizes the same extremal value $M_2=6-\log_222$. In general the bound need not be tight. We have verified exhaustively for $n=4,5$ that the bound is not attainable \footnote{Uniqueness in dimension three and non-attainability for $n\ge4$ are classical in the theory of trilinear forms. Over $\mathbb F_2$ an alternating trilinear form has a singular line unless $n=3$, where the unique exception is the determinant. This is the theorem of Draisma and Shaw~\cite{DraismaShaw10}, which applies to finite fields via the Chevalley--Warning theorem and builds on the classical classification of vector cross products in dimensions three and seven~\cite{Hurwitz1898, Eckmann1942, BrownGray67, Massey83}. To our knowledge this is the first translation of that classification into the extremization of the stabilizer R\'enyi entropy and the explicit bound above.}.

\section{Permutation-invariant Codes and Dicke States}
\label{Dicke}

A permutation-invariant code has a code space closed under arbitrary qubit permutations \cite{Ouyang14PRA}. Its canonical states are the Dicke states \cite{Dicke54PR}
\begin{equation}\label{eq:dicke-state}
  \ket{D^w_n}=\binom nw^{-1/2}\sum_{|x|=w}\ket x, \quad w=0,\cdots,n
\end{equation}
labelled by the excitation number $w$. With $\Lambda_w=\{x\in\mathbb F_2^n:|x|=w\}$ the weight-$w$ level of the cube, $\ket{D^w_n}=\ket{\Lambda_w}$ is a subset state and the indicator dictionary \eqref{eq:indicator} applies directly,
\begin{equation}\label{eq:dicke-dict}
  M_2(\ket{D^w_n})=4\log_2\binom nw-\log_2E^{(1)}(\Lambda_w).
\end{equation}
Writing $S_x=\Lambda_w\cap(\Lambda_w\oplus x)$ for the slice at $x$, what remains is to evaluate $\sum_xE(S_x)$ on one weight level.

$S_x$ is empty unless $|x|$ is even, since $|x|\equiv|z|+|z\oplus x|\equiv0 \pmod 2$ for $z\in S_x$. For $|x|=2j$ the identity $|z\oplus x|=|z|+2j-2|z\cap x|$ forces $|z\cap x|=j$, so a string paired by $x$ takes exactly $j$ of the $2j$ ones of $x$. Splitting the coordinates into $\mathrm{supp}(x)$ and its complement gives
\begin{equation}
  S_x\cong\Lambda^{(2j)}_{j}\times\Lambda^{(n-2j)}_{w-j},
  \quad \Lambda^{(s)}_r=\{y\in\mathbb F_2^s:|y|=r\},
\end{equation}
of size $A(x)=\binom{2j}{j}\binom{n-2j}{w-j}$. $S_x$ is nonempty only for $j\le\min(w,n-w)$, the range in which the complement can still supply $w-j$ ones.

The additive energy factorizes over products, $E(A\times B)=E(A)E(B)$, so $E(S_x)$ is a product of two level energies, and
\begin{equation}\label{eq:level}
  E\bigl(\Lambda^{(s)}_r\bigr) =\sum_{t=0}^{\min(r,\lfloor s/2\rfloor)}
   \binom s{2t}\binom{2t}{t}^2\binom{s-2t}{r-t}^2,
\end{equation}
in which a direction of weight $2t$ (there are $\binom s{2t}$ of them) contributes the squared slice size $\left[\binom{2t}{t}\binom{s-2t}{r-t}\right]^2$. The size $A(x)$ of a slice is not its energy, and only the energy enters \eqref{eq:dicke-dict}. Collecting the $\binom n{2j}$ directions of weight $2j$,
\begin{equation}\label{eq:dicke}
\begin{aligned}
  M_2(\ket{D^w_n})&=4\log_2\binom nw-\log_2E^{(1)}(\Lambda_w),\\
  E^{(1)}(\Lambda_w)&=\sum_{j=0}^{\min(w,n-w)}\binom n{2j}E\bigl(\Lambda^{(2j)}_j\bigr)E\bigl(\Lambda^{(n-2j)}_{w-j}\bigr).
\end{aligned}
\end{equation}
Eq.~\eqref{eq:dicke} is invariant under $w\to n-w$, since $X^{\otimes n}$ maps $\Lambda_w$ to $\Lambda_{n-w}$ and is Clifford.

We verify the formula in specific cases. The endpoints $w=0,n$ give $E^{(1)}=1$ and $M_2=0$, as they must for a computational basis state. For $w=1$, Eq.~\eqref{eq:level} gives $E\bigl(\Lambda^{(s)}_1\bigr)=3s^2-2s$, hence $E^{(1)}(\Lambda_1)=7n^2-6n$ and
\begin{equation}\label{eq:dicke-W}
  M_2(\ket{W_n})=3\log_2n-\log_2(7n-6),
\end{equation}
the value known for the $W$ state \cite{Odavic23SciPost,Catalano24arXiv}. For $w=2,3$ the sum \eqref{eq:dicke} evaluates in closed form,
\begin{equation}\label{eq:dicke23}
\begin{aligned}
  E^{(1)}(\Lambda_2)&=\frac14\,n(n-1)\bigl(91n^2-427n+492\bigr),\\
  E^{(1)}(\Lambda_3)&=\frac1{36}\,n(n-1)(n-2)\\
  &\quad\times\bigl(1645n^3-18921n^2+71708n-89244\bigr),
\end{aligned}
\end{equation}
reproducing the closed forms of Ref.~\cite{Liu25TIMPS}, obtained there by an MPS polynomial fit that extends step by step in $w$. 

Compared to previous results, Eq.~\eqref{eq:dicke} covers all $w$ at once. The sum in Eq.~\eqref{eq:dicke} has $\min(w,n-w)+1$ terms in place of the $4^n$ Pauli operators of Eq.~\eqref{eq:M2}. For example
\begin{equation}\label{eq:dicke-numbers}
  M_2(\ket{D^{10}_{100}})=44.53,\qquad M_2(\ket{D^{20}_{200}})=85.73,
\end{equation}
evaluated from at most eleven binomial terms, against $4^{100}$ and $4^{200}$ operators.

The derivation assumes only that the state has a single weight level. Permutation symmetry enters in organizing the sum by difference weight, and no symmetric sector or symmetric Pauli representative is used, in contrast to Ref.~\cite{Passarelli24PRA}. A general permutation-invariant code state, however, is a superposition of different excitation levels, which is not a subset state. 

\begin{figure}
 \centering
 \includegraphics[width=8.6cm]{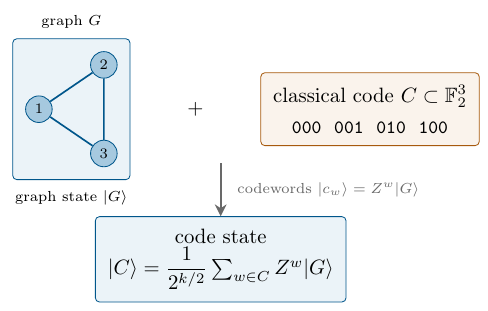}
 \caption{A CWS code is specified by a graph $G$ and a classical code $C$ (here $C\subseteq\mathbb F_2^3$ is shown); the code state $|C\rangle=2^{-k/2}\sum_{w\in C}Z^w|G\rangle$ is the uniform superposition of the codewords $|c_w\rangle=Z^w|G\rangle$. Each codeword is a stabilizer state, so the entire nonstabilizerness of the code state resides in the superposition.}
 \label{Fig2}
\end{figure}

\section{CWS code}
\label{sec:CWS}

\subsection{CWS Codes and the Classical Structure}
\label{sec:setting}

Consider a CWS code $(G,C)$, defined by a graph $G=(V, E)$ on $n$ vertices with adjacency matrix $\Gamma$ and a classical binary code $C\subseteq\mathbb{F}_2^n$ of size $m=2^k$ \cite{Cross09TIT, Chuang09JMP}. The codewords are $|c_w\rangle=Z^w|G\rangle$, $w\in C$, where $|G\rangle=\prod_{(u,v)\in E}\mathrm{CZ}_{uv}|+\rangle^{\otimes n}$ is the graph state. The construction is summarized in Fig.~\ref{Fig2}. A CWS code is a stabilizer code if and only if its classical code $C$ is an affine subspace of $\mathbb{F}_2^n$, a coset of a linear subspace \cite{Cross09TIT, Dehaene03PRA}. A non-affine choice of $C$ places the code outside the stabilizer family.

A central simplification is that the nonstabilizerness of an encoded state depends on $C$ alone (independent of the graph). Recall that the graph state is constructed by placing controlled-$Z$ gates on every edge of $G$ acting on the product state $|+\rangle^{\otimes n}$ \cite{Briegel01Jan, Raussendorf01May, Hein04PRA}:
\begin{equation}
 |G\rangle = \prod_{(i,j)\in E} \mathrm{CZ}_{ij}\,|+\rangle^{\otimes n}.
\end{equation}
Now apply a global Hadamard rotation $H^{\otimes n}$. Since $|+\rangle = H|0\rangle$, we can rewrite $|G\rangle = \bigl(\prod_{e\in E}\mathrm{CZ}_e\bigr)H^{\otimes n}|0\rangle^{\otimes n}$, and consequently
\begin{equation}
 H^{\otimes n}|G\rangle = \prod_{(i,j)\in E} \bigl(H^{\otimes n}\mathrm{CZ}_{ij}H^{\otimes n}\bigr)\,|0\rangle^{\otimes n}.
\end{equation}
For a single edge identified by the vertex-pair $(i,j)$, the Hadamards on qubits $k\notin\{i,j\}$ commute past $\mathrm{CZ}_{ij}$ and cancel, leaving a two-qubit Clifford gate:
\begin{equation}
 H^{\otimes n}\mathrm{CZ}_{ij}H^{\otimes n} = (H_i\otimes H_j)\,\mathrm{CZ}_{ij}\,(H_i\otimes H_j) \equiv \tilde{U}_{ij}.
\end{equation}
The full conjugated circuit is therefore $\tilde{U}_G = \prod_{(i,j)\in E} \tilde{U}_{ij}$, a product of two-qubit Clifford gates, one per edge. $\tilde{U}_{ij}$ and $\tilde{U}_{kl}$ on disjoint edges act on independent qubit pairs and can be executed in parallel. The circuit depth is thus controlled by the maximum number of edges that share the same vertex, i.e., the graph's maximum degree $\Delta(G)$. We now state this result formally. The detailed proof is given in Appendix~\ref{sec:decoupling-proof}.

\begin{lemma}[Graph--code decoupling]
\label{lem:decouple}
Let $|\psi_{\mathrm{enc}}\rangle = \sum_{w\in C} \psi_w Z^w |G\rangle$ be any encoded state of a CWS code based on graph $G = (V,E)$, where $Z^w \equiv \bigotimes_{i=1}^n Z^{w_i}$. Then
 \begin{equation}\label{eq:decoupling}
 H^{\otimes n} |\psi_{\mathrm{enc}}\rangle = \tilde{U}_G \sum_{w\in C} \psi_w |w\rangle,
 \end{equation}
 where
 \begin{equation}\label{eq:tildeU}
 \tilde{U}_G = \prod_{(i,j)\in E} (H_i \otimes H_j) \,\mathrm{CZ}_{ij}\, (H_i \otimes H_j)
 \end{equation}
 is a Clifford circuit whose depth equals the edge-chromatic number $\chi'(G)$ of $G$. By Vizing's theorem \cite{Diestel2017}, $\chi'(G) \le \Delta(G) + 1$. For bounded-degree graphs ($\Delta = O(1)$, the universal case of any geometrically local architecture) $\tilde{U}_G$ has constant depth independent of the code length $n$.
\end{lemma}

This Lemma is the key to relating nonstabilizerness to a classical structure. Since the second-order SRE $M_2$ is a Clifford invariant that vanishes exactly on stabilizer states \cite{Leone22Feb, Xiao26Aug}, the Lemma reduces the nonstabilizerness of any encoded state to that of the corresponding subset superposition. We therefore benchmark the nonstabilizerness of the code with the nonstabilizerness of the uniform code state
\begin{equation}
 |C\rangle := 2^{-k/2}\sum_{w\in C}Z^w|G\rangle .
\label{eq:codestate}
\end{equation}

We focus on $|C\rangle$ because it is the image of the encoding isometry on $|\bar{+}\rangle^{\otimes k}$ (compared with the physical state $|+\rangle$, a bar is used in $|\bar{+}\rangle$ to denote that it is a logical state), the standard input for $X$-basis computation. With a stabilizer input and a monotone $M_2$ \cite{Leone22Feb, Leone24PRA}, any nonstabilizerness in $|C\rangle$ has to be provided by the encoder. This state is also canonical for the code space: being phase-flat and maximally coherent in the codeword basis, it dephases to the maximally mixed logical state $\Pi_C/2^k$, so $M_2(|C\rangle)$ is a property of $C$ alone. Since the individual codewords $Z^w|G\rangle$ are stabilizer states \cite{Hein04PRA}, the nonstabilizerness is carried entirely by the superposition, i.e., by the code space rather than by its codewords.

Note that $|C\rangle$ is a benchmark (which we justify later in Sec.~\ref{sec:quant-cost}), not the most magical state the code space can host. Changing the phases of the superposition while keeping the same support changes the nonstabilizerness. For a single logical qubit, $C=\{0,1\}$, the state $|C\rangle=|+\rangle$ has $M_2=0$ while $T|+\rangle$, which lies in the same span, has $M_2=\log_2(4/3)\approx0.415$. For the Sidon $4$-set $\{0,e_1,e_2,e_3\}\subseteq\mathbb{F}_2^3$, the phase-shifted state $\frac12(|000\rangle+|100\rangle+|010\rangle+\mathrm{e}^{\mathrm{i}\pi/4}|001\rangle)$ has $M_2\approx1.752$, against the unshifted value $M_2(|C\rangle)=6-\log_2 22\approx1.541$. 

\begin{table*}
\caption{The dictionary between the additive structure of the classical code $C$ and the quantum resources of the CWS code. The fourth moment $E^{(1)}(C)$ and the encoding cost satisfy $J\le E^{(1)}(C)$, and consequently the code-state nonstabilizerness $M_2$ and the Choi-state nonstabilizerness satisfy $M_2(\mathrm{Choi})\ge M_2$. The weighted row applies to an arbitrary code-space state $|\psi\rangle=\sum_w\psi_w|w\rangle$ and reduces to the second row for $\psi_w\equiv m^{-1/2}$.}
\begin{tabular}{llll}
\toprule
combinatorial object & quantum object & physical reading & where\\
\midrule
$A(x)=|C\cap(C\oplus x)|$ & correlator $c(x,0)$ & nontrivial Pauli support & Eq.~\eqref{eq:A}\\
$E^{(1)}(C)=\sum_x E(C\cap(C\oplus x))$ & $M_2=4k-\log_2 E^{(1)}(C)$ & nonstabilizerness (uniform code state) & Eq.~\eqref{eq:M2full}\\
$\mathcal E^{(1)}_{\psi}(C)=\sum_{x,u}|A_2^{\psi}(x,u)|^2$ & $M_2(|\psi\rangle)=-\log_2\mathcal E^{(1)}_{\psi}(C)$ & nonstabilizerness (any code-space state) & Eq.~\eqref{eq:M2gen}\\
$J$ (encoding-sensitive) & $M_2(\mathrm{Choi})=4k-\log_2 J$ & encoding cost & Eq.~\eqref{eq:J}\\
\bottomrule
\end{tabular}
\label{tab:dict}
\end{table*}

Let $C\subseteq\mathbb{F}_2^n$ have size $m$, and let $|\psi\rangle=\sum_{w\in C}\psi_w|w\rangle$ be any normalized state of the coordinate subspace $\operatorname{span}\{|w\rangle \mid w\in C\}$. Define the correlator
\begin{equation}
 c_\psi(x,z):=\langle\psi|X^xZ^z|\psi\rangle
 =\sum_{w\in C\cap(C\oplus x)}\overline{\psi_{w\oplus x}}\,\psi_w\,(-1)^{z\cdot w},
\label{eq:corrgen}
\end{equation}
and the amplitude-weighted parallelogram energy
\begin{equation}
\begin{aligned}
 A_2^{\psi}(x,u)&:=\sum_{\substack{w\\ w,\,w\oplus x,\,w\oplus u,\,w\oplus x\oplus u\,\in C}}
 \overline{\psi_{w\oplus x}}\,\psi_w\,\psi_{w\oplus x\oplus u}\,\overline{\psi_{w\oplus u}},\\
 \mathcal E^{(1)}_{\psi}(C)&:=\sum_{x,u}\bigl|A_2^{\psi}(x,u)\bigr|^2 .
\end{aligned}
\label{eq:E1psi}
\end{equation}

\begin{lemma}[The dictionary]
\label{lem:dict}
The order-$2$ SRE of $|\psi\rangle$ is
\begin{equation}
 M_2(|\psi\rangle)
 =n-\log_2\sum_{x,z}\bigl|c_\psi(x,z)\bigr|^4
 =-\log_2\mathcal E^{(1)}_{\psi}(C).
\label{eq:M2gen}
\end{equation}
For the code state $|C\rangle=m^{-1/2}\sum_{w\in C}|w\rangle$ with $m=2^k$, the correlator reduces to $c_\psi(x,z)=c(x,z)$ (cf. Eq.~\eqref{eq:corr}). Consequently we have $A_2^{\psi}(x,u)=m^{-2}A_2(x,u)$, $\mathcal E^{(1)}_{\psi}(C)=m^{-4}E^{(1)}(C)$, and 
\begin{equation}
 M_2(|C\rangle)=4k-\log_2 E^{(1)}(C).
\label{eq:M2full}
\end{equation}
In either form the SRE $M_2$ depends on $C$ and the amplitudes $\psi$ alone, not on the graph $G$ or the physical length $n$.
\end{lemma}

\begin{proof}
By Lemma~\ref{lem:decouple} and the Clifford invariance of $M_2$, it suffices to compute for $|\psi\rangle=\sum_{w\in C}\psi_w|w\rangle$. For fixed $x$, define $f_x(w):=\overline{\psi_{w\oplus x}}\psi_w$ on the the set $S_x=C\cap(C\oplus x)$ and $f_x(w)=0$ outside it; then $c_\psi(x,z)=\widehat f_x(z)$, the Walsh--Hadamard transform of $f_x$. The fourth-moment form of Parseval equality
\[
 \sum_z\bigl|\widehat f_x(z)\bigr|^4
 =2^n\sum_u\Bigl|\textstyle\sum_w f_x(w)\overline{f_x(w\oplus u)}\Bigr|^2
\]
becomes, with $\sum_w f_x(w)\overline{f_x(w\oplus u)}=A_2^{\psi}(x,u)$,
\[
 \sum_z\bigl|c_\psi(x,z)\bigr|^4=2^n\sum_u\bigl|A_2^{\psi}(x,u)\bigr|^2 .
\]
Summing over $x$ gives $\sum_{x,z}|c_\psi(x,z)|^4=2^n\,\mathcal E^{(1)}_{\psi}(C)$, hence $M_2(|\psi\rangle)=n-\log_2\bigl(2^n\mathcal E^{(1)}_{\psi}(C)\bigr)=-\log_2\mathcal E^{(1)}_{\psi}(C)$. For $\psi_w\equiv m^{-1/2}$ each summand of $A_2^{\psi}(x,u)$ equals $m^{-2}$, so $A_2^{\psi}(x,u)=m^{-2}A_2(x,u)$ and $\mathcal E^{(1)}_{\psi}(C)=m^{-4}\sum_{x,u}A_2(x,u)^2=m^{-4}E^{(1)}(C)$, giving $M_2(|C\rangle)=4\log_2 m-\log_2 E^{(1)}(C)=4k-\log_2 E^{(1)}(C)$.
\end{proof}

Writing $\psi_w=\sqrt{p_w}\,\mathrm e^{\mathrm i\phi_w}$, each summand of $A_2^{\psi}(x,u)$ has magnitude $\sqrt{p_wp_{w\oplus x}p_{w\oplus u}p_{w\oplus x\oplus u}}$ and phase $\phi_w-\phi_{w\oplus x}-\phi_{w\oplus u}+\phi_{w\oplus x\oplus u}$, the discrete second difference of $\phi$ across the parallelogram $\{w,w\oplus x,w\oplus u,w\oplus x\oplus u\}$. The amplitudes thus weight the parallelograms of $C$, while the phases enter only through this curvature. For the flat code state the curvature vanishes identically, so $M_2(|C\rangle)$ depends only on the support $C$.

The nonstabilizerness of $|C\rangle$ is determined by the parallelogram energy $E^{(1)}(C)$ alone, introducing the toolbox of combinatorics to the problem. Table~\ref{tab:dict} collects the dictionary. By Eq.~\eqref{eq:M2full}, $M_2$ decreases with $E^{(1)}(C)$, so the more additive structure $C$ has, the less nonstabilizerness it carries. The zero-nonstabilizerness end is the affine case, where $E^{(1)}(C)=2^{4k}$; the maximal-nonstabilizerness end corresponds to the sparsest difference structure, namely the Sidon set case, which we turn to next.

The dictionary extends to every R\'enyi order. Replacing the fourth moment by the $2\alpha$-th moment in the computation above yields
\begin{equation}
\begin{aligned}
 M_\alpha(|C\rangle)&=\frac{2\alpha k-\log_2 E^{(\alpha-1)}(C)}{\alpha-1},\\
 E^{(\alpha-1)}(C):&=\sum_x E_\alpha\bigl(C\cap(C\oplus x)\bigr),
\end{aligned}
\label{eq:Malpha}
\end{equation}
where $E_\alpha$ is the order-$2\alpha$ energy. The derivation and the corresponding discussion about coset-invariance are given in Appendix~\ref{sec:appalpha}.

\begin{figure*}
 \centering
 \includegraphics[width=12.9cm]{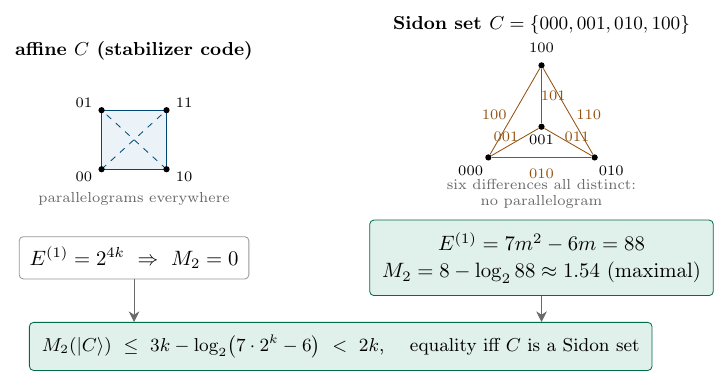}
 \caption{Extremes. Left: an affine $C$ corresponds a stabilizer code; the whole set is an affine $2$-plane with parallelograms everywhere. Right: the Sidon set
 $C=\{000,001,010,100\}\subseteq\mathbb F_2^3$; its six pairwise
 differences are all distinct, so $A(x)\in\{0,2\}$ and $C$ contains no
 parallelogram. Sidon sets attain the universal bound
 $M_2\le 3k-\log_2(7\cdot 2^k-6)$.}
 \label{Fig3}
\end{figure*}

\subsection{Extremal Nonstabilizerness}
\label{sec:quant-cost}

Equation \eqref{eq:M2full} turns the question of ``how much nonstabilizerness'' to be a minimization problem over the parallelogram energy $E^{(1)}(C)$. This subsection discusses two aspects of this problem. Firstly, we minimize $E^{(1)}(C)$ over all codes of a fixed size, determining which code is most magical. Subsequently, we fix a code and minimize the weighted energy of Lemma~\ref{lem:dict} over the states of its code space, obtaining the most magical state that a code space can host.

\subsubsection{The code-state maximum}

The minimization of $E^{(1)}(C)$, and hence the maximization of nonstabilizerness, is governed by the difference structure of $C$. Since $a\in C\cap(C\oplus x)$ holds only when both $a$ and $a\oplus x$ lie in $C$, the multiplicity
\begin{equation}
 A(x)=|C\cap(C\oplus x)|=\#\bigl\{(a,b)\in C^2\mid a\oplus b=x\bigr\}
\end{equation}
counts the ordered pairs of codewords differing by $x$. Over $\mathbb{F}_2$ the difference is symmetric ($a\oplus b=b\oplus a$), so the two orientations of an unordered pair $\{a,b\}$ always occur together, i.e., for $x\ne0$,
\begin{equation}
 A(x)=2\,\#\bigl\{\{a,b\}\subseteq C\mid a\oplus b=x\bigr\}
\end{equation}
is an even number.

A set $C\subseteq\mathbb{F}_2^n$ is a \emph{Sidon set} \cite{Sidon1932, Babai1985101} if all its nonzero pairwise differences are distinct. The $\binom m2$ nonzero differences being distinct means that each $x\ne0$ is realized by at most one unordered pair, so $A(x)$ is $0$ or $2$; conversely, $A(x)\in\{0,2\}$ for every $x\ne0$ forbids two pairs from sharing a difference. Sidonness is therefore equivalent to
\begin{equation}
 A(x)\in\{0,2\}\quad\text{for all }x\ne0 .
\label{eq:sidoncrit}
\end{equation}
Recall that with $S_x=C\cap(C\oplus x)$, we have $|S_x|=A(x)$ and 
\begin{equation}
 E(C)=\sum_x A(x)^2=\sum_x|S_x|^2,\quad
 E^{(1)}(C)=\sum_x E(S_x).
\end{equation}
For $x=0$, we have $S_0=C$, which contributes $m^2$ to $E(C)$ and $E(C)$ to $E^{(1)}(C)$. For all $x\ne0$, $A(x)\in\{0,2\}$ indicates that $S_x$ is either empty or a two-element subset of $C$. Because the $\binom m2$ nonzero differences are all distinct, exactly $\binom m2$ values of $x\ne0$ give such a subset, one for each difference, and the remaining $x\ne0$ give the empty set. 

All of these subsets have the same additive energy,
\begin{equation}
 E(S_x)=8, \; x\ne0.
\label{eq:twoelem}
\end{equation}
To see this, write $S_x=\{a,a\oplus x\}$ and translate it by $a$, so that $S_x=\{0,x\}$. A quadruple $(v_1,v_2,v_3,v_4)\in S_x^4$ satisfies $v_1\oplus v_2\oplus v_3\oplus v_4=0$ if and only if an even number of $v_i$ equal $x$, which gives $\binom40+\binom42+\binom44=8$ quadruples. Equivalently, $|S_x\cap(S_x\oplus u)|$ equals $2$ for $u\in\{0,x\}$ and equals $0$ otherwise, so $E(S_x)=2^2+2^2=8$. The $\binom m2$ two-element sets thus contribute $8\binom m2=4m(m-1)$ to $E^{(1)}(C)$, and their sizes $A(x)=2$ contribute $2^2\binom m2=2m(m-1)$ to $E(C)$. Together with the contribution from the $x=0$ case, we have
\begin{equation}
\begin{aligned}
 E(C)&=m^2+2m(m-1)=3m^2-2m,\\
 E^{(1)}(C)&=E(C)+8\binom{m}{2}=7m^2-6m.
\end{aligned}
\label{eq:sidonclosed}
\end{equation}
Eq.~\eqref{eq:M2full} then gives $M_2(|C\rangle)$.

\begin{theorem}[Sidon set maximizes nonstabilizerness]
\label{thm:ext}
For every $C\subseteq\mathbb{F}_2^n$ of size $m=2^k$,
\begin{equation}
 M_2(|C\rangle) \le 3k-\log_2\bigl(7\cdot 2^k-6\bigr) < 2k.
\end{equation}
Equality holds if and only if $C$ is a Sidon set. 
\end{theorem}

A CWS code state $|C\rangle$ therefore carries less than twice its logical-qubit count of nonstabilizerness, independent of the physical size $n$. The two extremes---affine and Sidon codes---are illustrated in
Fig.~\ref{Fig3}. It remains to prove the inequality in Theorem~\ref{thm:ext}, which is given by two lower bounds on the energies. For the additive energy, separate out the $y=0$ term and write $A(y)=2b(y)$ for $y\ne0$, where $b(y)=\#\{\{a,b\}\subseteq C\mid a\oplus b=y\}$ is the number of unordered pairs of codewords with difference $y$. Using $b(y)^2\ge b(y)$ for the nonnegative integer $b(y)$ (equality holds iff $b(y)\in\{0,1\}$ for every $y\ne0$, i.e., iff $C$ is Sidon), we have
\begin{equation}
 \begin{aligned}
 E(C)&=m^2+\sum_{y\ne0}A(y)^2=m^2+4\sum_{y\ne0}b(y)^2\\
 &\ge m^2+4\sum_{y\ne0}b(y)=m^2+4\binom m2\\
 &=3m^2-2m,
 \end{aligned}
\end{equation}
For the parallelogram energy, note that $E(S_x)\ge4A(x)$ for every $x\ne0$. The empty case is trivial, a two-element $S_x$ has $E(S_x)=8=4\cdot2$, and for $A(x)\ge4$ we have
\begin{equation}
 \begin{aligned}
 E(S_x)&=\sum_u|S_x\cap(S_x\oplus u)|^2\\ &\ge\sum_u|S_x\cap(S_x\oplus u)|=|S_x|^2=A(x)^2\\
 &\ge4A(x).
 \end{aligned}
\end{equation}
Combining the two inequalities gives rise to
\begin{equation}
\begin{aligned}
 E^{(1)}(C)&=E(C)+\sum_{x\ne0}E(S_x)\\
 &\ge(3m^2-2m)+4\sum_{x\ne0}A(x)\\
 &=(3m^2-2m)+4m(m-1)=7m^2-6m.
\end{aligned}
\end{equation}
Equality holds if and only if $C$ is Sidon. Application of Eq.~\eqref{eq:M2full} then yields the
theorem.

The bound $M_2(|C\rangle) \le 3k-\log_2\bigl(7\cdot 2^k-6\bigr)$ is tight whenever a Sidon set exists. A Sidon set of size $m$ has $\binom m2$ distinct nonzero differences, so it requires $\binom m2\le2^n-1$; conversely, Sidon sets of size $2^k$ exist in $\mathbb{F}_2^n$ for all $k\le\lfloor n/2\rfloor$, for instance $\{(x,x^3) \mid x\in\mathbb{F}_2^s\}\subseteq\mathbb{F}_2^s\times\mathbb{F}_2^s$ at even $n=2s$ \cite{Czerwinski24AdvMathCommun}. Hence for every $k\le\lfloor n/2\rfloor$ the inequality is saturated. 

When $\binom m2>2^n-1$ no Sidon set exists, and the maximizer is a near-Sidon set whose difference multiplicities take a small set of values. This regime is technically richer but the physical intuition is direct. We summarize the results here and give details in Appendix~\ref{sec:appunsat}. The essential point is that the energies admit closed forms in terms of the affine $2$-planes of $C$,
\begin{equation}
\begin{aligned}
 E(C)&=3m^{2}-2m+24N,\\
 E^{(1)}(C)&=7m^{2}-6m+72N+96Q,
\end{aligned}
\label{eq:planes}
\end{equation}
where $N$ counts the affine $2$-planes contained in $C$ and $Q$ is the weighted count by 2-subspaces, so the extremal problem reduces to a finite integer program. In the half-space regime $n=k+1$ the optimum can be solved by dynamic programming. For example, $(n,k)=(4,3)$ has maximal $M_2=12-\log_2 1240\approx1.724$, and as $k\to\infty$ the maximum is $M_2=k+1-\log_2 7+o(1)$, with half the Sidon slope of $2$.

\subsubsection{The code-space capacity}\label{sec:codecapacity}

We now maximize $M_2$ over the code space. Let
\begin{equation}
 V_C:=\operatorname{span}\bigl\{Z^w|G\rangle \mid w\in C\bigr\}
\end{equation}
be the code space, and define its nonstabilizerness capacity as
\begin{equation}
 M_2^{\max}(C):=\max_{|\psi\rangle\in V_C}M_2(|\psi\rangle).
\label{eq:capacity}
\end{equation}
By Lemma~\ref{lem:decouple} and the Clifford invariance of $M_2$ \cite{Leone22Feb, Xiao26Aug}, the maximization is the same as over the coordinate subspace $\operatorname{span}\{|w\rangle \mid w\in C\}\cong\mathbb C^{2^k}$. The weighted dictionary of Lemma~\ref{lem:dict} then turns this into a variational problem: with the amplitudes $p_w=|\psi_w|^2$ and phases $\phi_w=\arg\psi_w$ ranging independently over the simplex $\sum_w p_w=1$, we have
\begin{equation}
 M_2^{\max}(C)=-\log_2\Bigl[\min_{\psi:\ \sum_w|\psi_w|^2=1}\mathcal E^{(1)}_{\psi}(C)\Bigr].
\label{eq:capmin}
\end{equation}
The code state $|C\rangle$ corresponds to the flat point $\psi_w\equiv m^{-1/2}$, whose value $4k-\log_2 E^{(1)}(C)$ of Eq.~\eqref{eq:M2full} is the benchmark. The two solvable endpoints of this minimization (Sidon and affine codes) are discussed below.

For every CWS code $(G,C)$ with $|C|=2^k$, every state of the code space satisfies
\begin{equation}\label{eq:capbound}
 M_2(|\psi\rangle)\le 2k ,
\end{equation}
hence $M_2^{\max}(C)\le2k$ independent of the physical block length $n$. Indeed, by Lemma~\ref{lem:decouple} it suffices to bound $|\psi\rangle=\sum_{w\in C}\psi_w|w\rangle$ with $\sum_w|\psi_w|^2=1$. With $p_w=|\psi_w|^2$, the $x=0$ slice of the Pauli spectrum is $\langle\psi|Z^z|\psi\rangle=\sum_w p_w(-1)^{z\cdot w}=\hat p(z)$. Parseval equality and Cauchy--Schwarz inequality give
\begin{equation}
\begin{aligned}
 \sum_{x,z}|c(x,z)|^4 & \ge\sum_z|\hat p(z)|^4
 \ge2^{-n}\Bigl(\sum_z|\hat p(z)|^2\Bigr)^2\\
 &=2^n\Bigl(\sum_w p_w^2\Bigr)^2\ge2^{n-2k},
\end{aligned}
\end{equation}
so $M_2(|\psi\rangle)\le n-(n-2k)=2k$. The bound is optimal up to $O(1)$: for a Sidon code it equals $2k-\log_2 6+O(2^{-k})$, as the next theorem shows, while a generic $n$-qubit state has $M_2\approx n$ \cite{Szombathy25PRR}.

The Pauli fourth moment of a general code-space state factorizes through the codeword amplitudes and phases, and the amplitude optimization is solvable. Both steps are detailed in Appendix~\ref{sec:appcapacity}, which yield the following Theorem.

\begin{theorem}[Nonstabilizerness capacity of a Sidon code]
\label{thm:capacity}
Let $C\subseteq\mathbb{F}_2^n$ be a Sidon set of size $m=2^k\ge4$. Then
\begin{equation}
 M_2^{\max}(C)=3k-\log_2\bigl(6\cdot2^k-6\bigr)=2k-\log_2 6+O(2^{-k}) ,
\end{equation}
attained by every phase-balanced state $|\psi\rangle=m^{-1/2}\sum_{w\in C}\mathrm{e}^{\mathrm i\phi_w}|w\rangle$ with $\sum_w \mathrm{e}^{4\mathrm i\phi_w}=0$. 
\end{theorem}

By Appendix~\ref{sec:appcapacity}, $\min_\psi\sum_{x,z}|c(x,z)|^4=2^n\cdot6(m-1)/m^3$, with the phase term vanishing at the balanced phases, so $M_2^{\max}(C)=n-\log_2\bigl[2^n 6(m-1)/m^3\bigr]=3k-\log_2(6\cdot2^k-6)$. The code state itself has $M_2(|C\rangle)=3k-\log_2(7\cdot2^k-6)$ by Eq.~\eqref{eq:sidonclosed}, so phase optimization over a Sidon code space raises the nonstabilizerness by only $\log_2(7/6)+O(2^{-k})\approx0.22$ bits, i.e., the flat benchmark is \emph{nearly optimal}. For $k=1$ the unique two-element Sidon set is affine and $M_2^{\max}(C)=\log_2(3/2)$. 

For affine codes the situation is on the opposite side. The affine support can be bijectively mapped to the coordinate space $\langle e_1,\dots,e_k\rangle$. The code space is thus equivalent to the $k$-qubit Hilbert space.

\begin{proposition}
\label{prop:affinecap}
If $C$ is an affine space of dimension $k$, then $M_2^{\max}(C)=\mu_k$, where
\begin{equation}
 \mu_k:=\max_{|\phi\rangle\in(\mathbb C^2)^{\otimes k}}M_2(|\phi\rangle)
\end{equation}
is the maximal SRE of a $k$-qubit state, independent of $n$ and of the graph.
\end{proposition}

The values of $\mu_k$ are known for small $k$, e.g., $\mu_1=\log_2(3/2)$, $\mu_2=\log_2(16/7)$ \cite{Liu26QST}. The symmetric informationally complete (SIC) bound \cite{Cuffaro24arXiv}
\begin{equation}
 \mu_k\le\log_2\frac{2^k+1}{2}=k-1+\log_2\bigl(1+2^{-k}\bigr)<k
\end{equation}
gives $\mu_k \le k-O(1)$. An affine code space thus hosts up to $\mu_k\approx k$ bits of nonstabilizerness although its code state $|C\rangle$ is a stabilizer state with $M_2(|C\rangle)=0$. The flat benchmark and the capacity are separated by an unbounded gap in $k$, in contrast to the $O(1)$ gap of the Sidon case.

\subsubsection{The benchmark and its stability}

The following proposition locates the code state $|C\rangle$ within the code space and justifies $M_2(|C\rangle)$ as a benchmark.

\begin{proposition}
\label{prop:flatmin}
For every code, among the phase-flat states $|\psi\rangle=m^{-1/2}\sum_{w\in C}\mathrm{e}^{\mathrm i\phi_w}|w\rangle$ the code state $|C\rangle$ minimizes $M_2$:
\begin{equation}
 M_2(|\psi\rangle)\ge M_2(|C\rangle).
\end{equation}
The equality holds if and only if, on every difference $S_x=C\cap(C\oplus x)$, the relative phase $\mathrm{e}^{\mathrm i(\phi_{w\oplus x}-\phi_w)}$ is a constant.
\end{proposition}

\begin{proof}
For a phase-flat state the correlator reads $c_\psi(x,z)=m^{-1}\sum_{w\in S_x}\mathrm{e}^{\mathrm i(\phi_{w\oplus x}-\phi_w)}(-1)^{z\cdot w}$, i.e.\ $m^{-1}$ times the Walsh transform of a unit-modulus function on $S_x$, while $c(x,z)=m^{-1}\sum_{w\in S_x}(-1)^{z\cdot w}$ corresponds to the constant unit function. The fourth moment of a Walsh transform is maximized by the constant function: for unit $|f(w)|=1$ and every $u$, $|\sum_w f(w)\overline{f(w\oplus u)}|\le|S_x\cap(S_x\oplus u)|$, with equality given by $f\equiv1$, so $\sum_z|c_\psi(x,z)|^4\le\sum_z|c(x,z)|^4$ for each $x$. Summing over $x$ gives the claim. 
\end{proof}

For $(x,u)\in\mathbb{F}_2^n\times\mathbb{F}_2^n$, let
\begin{equation}
 \mathcal P_{x,u}:=\bigl\{w\in C \mid w, w\oplus x, w\oplus u, w\oplus x\oplus u\in C\bigr\},
\end{equation}
so that $A_2(x,u)=|\mathcal P_{x,u}|$. Write an arbitrary code-space state as
\begin{equation}
 |\psi\rangle=\sum_{w\in C}m^{-1/2}\sqrt{1+\varepsilon_w}\;\mathrm e^{\mathrm i\phi_w}|w\rangle,
 \qquad \textstyle\sum_w\varepsilon_w=0,
\end{equation}
and define the curvature of the phase function along the parallelogram anchored at $w$ as
\begin{equation}
 \kappa_{x,u}(w):=\phi_w-\phi_{w\oplus x}-\phi_{w\oplus u}+\phi_{w\oplus x\oplus u}.
\end{equation}
Then the gap between $M_2(|\psi\rangle)$ and $M_2(|C\rangle)$ can be evaluated by decomposing the nonstabilizerness shift into a phase part and an amplitude part. 

\begin{proposition}[Stability of the benchmark]
\label{prop:stability}
The nonstabilizerness deviation of a code-space state from the benchmark is
\begin{equation}
 M_2(|\psi\rangle)-M_2(|C\rangle)=\log_2\frac{E^{(1)}(C)}{m^{4}\,\mathcal E^{(1)}_{\psi}(C)} .
\label{eq:shift}
\end{equation}
It separates into a phase and an amplitude part.

\noindent (i) Phases raise the nonstabilizerness quadratically. For $\varepsilon_w\equiv0$, we have
\begin{equation}
 M_2(|\psi\rangle)-M_2(|C\rangle)
 =-\log_2\Bigl(1-\frac{\Xi(\phi)}{E^{(1)}(C)}\Bigr)\ge0 ,
\label{eq:phase-exact}
\end{equation}
where
\begin{equation}
 \Xi(\phi):=\sum_{x,u}\sum_{\substack{w,w'\in\mathcal P_{x,u}\\ w\ne w'}}
 \Bigl[1-\cos\bigl(\kappa_{x,u}(w)-\kappa_{x,u}(w')\bigr)\Bigr]\ge0 ,
\end{equation}
and to leading order the shift is quadratic in the phases,
\begin{equation}
\begin{aligned}
 &M_2(|\psi\rangle)-M_2(|C\rangle)\\
 ={}&\frac{1}{2\ln2} \frac{\displaystyle\sum_{x,u}\sum_{w,w'\in\mathcal P_{x,u}}\bigl(\kappa_{x,u}(w)-\kappa_{x,u}(w')\bigr)^2}{E^{(1)}(C)}
 +O(|\phi|^4) .
\end{aligned}
\label{eq:phase-quad}
\end{equation}

\noindent (ii) Amplitude imbalance: first order in general and second order on the extremal families. For $\phi_w\equiv0$, we have
\begin{equation}
 \mathcal E^{(1)}_{\psi}(C)=m^{-4}\Bigl[E^{(1)}(C)+\textstyle\sum_w\varepsilon_w\,G(w)+O(|\varepsilon|^2)\Bigr],
\label{eq:amp-exp}
\end{equation}
with the vertex weight
\begin{equation}
\begin{aligned}
 G(w):&=\sum_{x,u}A_2(x,u)\,N_{x,u}(w),\\
 N_{x,u}(w):&=\sum_{w'\in\mathcal P_{x,u}}\;\sum_{v\in\{w',\,w'\oplus x,\,w'\oplus u,\,w'\oplus x\oplus u\}}\mathbf 1_{v=w} .
\end{aligned}
\end{equation}
For affine codes $G(w)=4m^3$ and for Sidon codes $G(w)=28m-24$, both independent of $w$, so the benchmark is a critical point of $\mathcal E^{(1)}_{\psi}$ in the amplitude direction and the imbalance acts only at second order. For intermediate codes $G$ is non-constant (e.g., the non-Sidon set of Appendix~\ref{sec:appunsat} has $G\in\{536,704\}$), and the imbalance enters $M_2$ at first order.
\end{proposition}

\begin{proof}
Equation~\eqref{eq:shift} follows from $M_2(|\psi\rangle)=-\log_2\mathcal E^{(1)}_{\psi}(C)$ (Lemma~\ref{lem:dict}) and $M_2(|C\rangle)=-\log_2\bigl(m^{-4}E^{(1)}(C)\bigr)$.

(i) For flat amplitudes $A_2^{\psi}(x,u)=m^{-2}\sum_{w\in\mathcal P_{x,u}}\mathrm e^{\mathrm i\kappa_{x,u}(w)}$, so $m^{4}\,\mathcal E^{(1)}_{\psi}(C) =\sum_{x,u}\Bigl|\textstyle\sum_{w\in\mathcal P_{x,u}}\mathrm e^{\mathrm i\kappa_{x,u}(w)}\Bigr|^2 =E^{(1)}(C)-\Xi(\phi)$, and \eqref{eq:phase-exact} follows. Expanding $1-\cos\theta=\tfrac12\theta^2+O(\theta^4)$ gives \eqref{eq:phase-quad}. Since $\Xi\ge0$, this is the quantitative form of Proposition~\ref{prop:flatmin}.

(ii) Write $\psi_w=m^{-1/2}(1+\tfrac12\varepsilon_w)+O(\varepsilon^2)$ and expand $A_2^{\psi}(x,u)$ to first order; grouping the $\varepsilon_v$ term by the incidence of $v$ among the four vertices of each parallelogram and summing over $(x,u)$ gives \eqref{eq:amp-exp}. For affine $C$ the nonempty cases are $x,u\in C$, each with $A_2(x,u)=m$ and $N_{x,u}(w)=4$, so $G(w)=m^2\cdot m\cdot4=4m^3$. For Sidon $C$ the nonempty cases are the trivial case $(0,0)$ and the three pair cases $(x,0),(0,x),(x,x)$ for each of the $\binom m2$ differences. The first contributes $4m$ to each $w$, and each pair contributes $8$ to each of its two endpoints through each of the three types, giving $G(w)=4m+24(m-1)=28m-24$.
\end{proof}

Note that $|C\rangle$ is a phase-flat minimizer only, not a minimizer in the whole code space. Each codeword $Z^w|G\rangle\in V_C$ is a stabilizer state \cite{Hein04PRA} with $M_2(Z^w|G\rangle)=0<M_2(|C\rangle)$ for non-affine $C$. The benchmark is separated from the upper bound only through phases. By Theorem~\ref{thm:capacity} and Proposition~\ref{prop:affinecap}, phase optimization reaches the capacity, which exceeds $M_2(|C\rangle)$ by $\log_2(7/6)+O(2^{-k})$ for Sidon codes and by $\mu_k\approx k$ for affine codes. Eq.~\eqref{eq:phase-exact} exhibits this phase contribution as quadratic in the curvature. 

On a Sidon code, the flat state maximizes $M_2$ over all phase-flat states. The only nonempty cases are the trivial case and the pair cases, so $\mathcal E^{(1)}_{\psi}=7s_2^2-6s_4$ with $s_q=\sum_w p_w^q$, whose minimum $(7m-6)/m^3$ is attained at $p_w\equiv m^{-1}$. Amplitude imbalance can therefore only lower a Sidon code state's nonstabilizerness below the benchmark, in contrast to the affine code where the benchmark is a stabilizer state at the bottom. Thus $M_2(|C\rangle)$ is the lower baseline of the code space and $M_2^{\max}(C)$ the upper one.

\subsubsection{Choi state: encoding-sensitive nonstabilizerness}
\label{sec:choi}

The Choi state encodes the full action of the encoding isometry \cite{Watrous18book, Wilde17book}. A natural question is whether its nonstabilizerness agrees with the code's nonstabilizerness. We show that it is never smaller, and the gap measures how much the nonstabilizerness depends on the labeling.

Let $|\Psi\rangle_{AP}=2^{-k/2}\sum_{i\in\mathbb{F}_2^k}|i\rangle_A\otimes Z^{w_i}|G\rangle_P$ be the Choi state, where $w:\mathbb{F}_2^k\to C$, $i\mapsto w_i$ is a bijection labeling the codewords. The register $A$ represents the $k$-qubit logical input space before encoding, while $P$ represents the $n$-qubit physical code space after encoding. By the graph--code decoupling Lemma, $|\Psi\rangle$ is Clifford-equivalent to the subset state $|\Phi\rangle=2^{-k/2}\sum_i|i\rangle_A|w_i\rangle_P$ of the graph $S:=\{(i,w_i)\mid i\in\mathbb{F}_2^k\}\subseteq\mathbb{F}_2^{k+n}$. Since $|S|=k$, we have
\begin{equation}\label{eq:J}
 M_2(|\Psi\rangle)=4k-\log_2 J, \quad J:=E^{(1)}(S).
\end{equation}

\begin{proposition}
\label{prop:Jle}
For any bijective encoding $w:\mathbb{F}_2^k\to C$,
\begin{equation}
 J\le E^{(1)}(C),
 \quad\text{equivalently}\quad
 M_2(\mathrm{Choi})\ge M_2(\mathrm{code}).
\end{equation}
The Choi state always carries at least as much nonstabilizerness as the code state.
\end{proposition}

\begin{proof}

\emph{Parallelograms of $S$ and of $C$.}
Write the ambient offsets of $S$ as $x=(u,a)$ and $y=(v,b)$, with $u,v\in\mathbb{F}_2^k$ and $a,b\in\mathbb{F}_2^n$, and take the anchor $(i,w_i)$. The four points of the parallelogram are
\begin{equation}
 (i,w_i),\; (i\oplus u,w_i\oplus a),\;
 (i\oplus v,w_i\oplus b),\; (i\oplus u\oplus v,w_i\oplus a\oplus b).
\end{equation}
All of them lie in $S$ if and only if the physical quadruple $\{w_i,w_i\oplus a,w_i\oplus b,w_i\oplus a\oplus b\}$ lies in $C$ and the labeling reproduces it in the logical addresses,
\begin{equation}
\label{eq:labelclosure}
 w_{i\oplus u}=w_i\oplus a,\;
 w_{i\oplus v}=w_i\oplus b,\;
 w_{i\oplus u\oplus v}=w_i\oplus a\oplus b .
\end{equation}
Thus $A_2^S(x,y)$ counts those anchors of $(a,b)$ on which the labeling is closed, while $A_2(a,b)$ counts all anchors of $(a,b)$ in $C$. 

\emph{Partition by address.}
Fix $(a,b)$ and an anchor $w$ of $(a,b)$ in $C$, with logical address $i=\lambda(w)$. Because $\lambda$ is a bijection, the first two equations of \eqref{eq:labelclosure} leave no freedom,
\begin{equation}
 u=\lambda(w\oplus a)\oplus\lambda(w),\quad
 v=\lambda(w\oplus b)\oplus\lambda(w),
\end{equation}
and only the third survives as a condition,
\begin{equation}
\label{eq:labelflat}
 \lambda(w\oplus a\oplus b)=\lambda(w)\oplus\lambda(w\oplus a)\oplus\lambda(w\oplus b).
\end{equation}
The address $(u,v)$ is therefore a function of the anchor: each anchor of $(a,b)$ contributes to exactly one address class, or, if \eqref{eq:labelflat} fails, to none. Writing $n_{uv}(a,b)=A_2^S\bigl((u,v),(a,b)\bigr)$ for the class sizes, this is the statement that the anchors of $(a,b)$ are partitioned by address,
\begin{equation}
\label{eq:partition}
 \sum_{u,v}n_{uv}(a,b)=A_2^{\mathrm{cl}}(a,b)\le A_2(a,b),
\end{equation}
where $A_2^{\mathrm{cl}}$ counts the anchors on which the labeling is closed.

\emph{The inequality.}
Summing \eqref{eq:partition} squared over addresses and offsets, we obtain
\begin{equation}\label{eq:Jchain}
\begin{aligned}
  J&=\sum_{a,b}\sum_{u,v}n_{uv}(a,b)^2\\
  &\le\sum_{a,b}\Bigl(\sum_{u,v}n_{uv}(a,b)\Bigr)^2\\
  &\le\sum_{a,b}A_2(a,b)^2=E^{(1)}(C).
\end{aligned}
\end{equation}
Th equality $J=E^{(1)}(C)$ holds precisely when, for every $(a,b)$ with $A_2(a,b)>0$, all its anchors lie in a single address class on which $\lambda$ is closed.
\end{proof}

The gap has a direct reading,
\begin{equation}
 M_2(\mathrm{Choi})-M_2(\mathrm{code})=\log_2\frac{E^{(1)}(C)}{J}\ \ge 0 ,
\end{equation}
as the nonstabilizerness that the encoder owes to the address assignment alone and that the code state cannot reveal. Since $(x,y)=(0,0)$ contributes $m^2$ to $J$, one has $J\ge m^2$ and hence $M_2(\mathrm{Choi})\le 2k$, so the Choi state does not leave the range available to the code state.

$J$ is not a function of $C$ alone but also influenced by the encoding map $w$. Equality $J=E^{(1)}(C)$ holds for Sidon sets with any labeling and for affine codes with a linear labeling, but not in general. Among the $12870$ eight-element subsets of $\mathbb{F}_2^4$, there are $774$ non-Sidon sets with $J=E^{(1)}(C)$. The linear code $C=\mathbb{F}_2^3\subseteq\mathbb{F}_2^6$ has $E^{(1)}(C)=4096$ and $J=4096$ under the linear labeling but $J=736$ under the scrambled labeling $w=(0,1,2,4,3,6,5,7)$. A thorough characterization of the equality class remains open.

\begin{figure*}
 \centering
 \includegraphics[width=17.2cm]{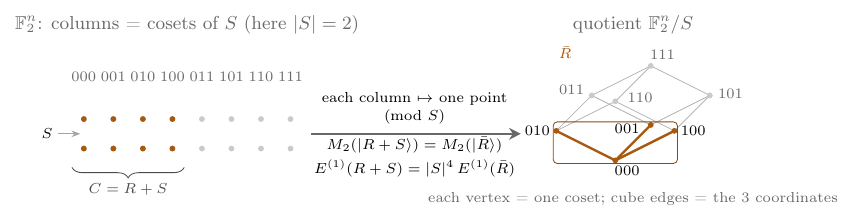}
 \caption{Coset invariance of nonstabilizerness. $\bar R=\{000,001,010,100\}$ is the Sidon $4$-set, $|S|=2$, and $k=3$. Left: the ambient space $\mathbb{F}_2^n$ is partitioned into the cosets of a subgroup $S$, one column per coset (the first column is $S$ itself); the coset closure $C=R+S$ is the union of the highlighted cosets. Right: the quotient $\mathbb{F}_2^n/S$, with each coset collapsed to a point. The quotient is drawn as a cube so that its coordinates remain visible: the three edge directions are the three coordinates of $\mathbb{F}_2^n/S$, and $\bar R$ is the star $\{000\}$ together with its three neighbours. Each face of the cube through $000$ is a parallelogram $\{000,x,y,x\oplus y\}$, and the Sidon condition $x\oplus y\notin\bar R$ means that every such face misses its diagonal vertex: $\bar R$ never completes a parallelogram. The energies scale as $E^{(1)}(R+S)=|S|^4 E^{(1)}(\bar R)$, and the factor $|S|^4$ cancels the shift $4\log_2|S|$ of the logical dimension in Eq.~\eqref{eq:M2full}, so that $M_2(|R+S\rangle)=M_2(|\bar R\rangle)$.}
 \label{Fig4}
\end{figure*}

\subsection{Coset Closure and the Kerdock Codes}
\label{sec:co-Ker}

\subsubsection{Nonstabilizerness depends only on the quotient}
\label{sec:coset}

A stabilizer code has no nonstabilizerness, so enlarging a code by a stabilizer-type subgroup introduces no more nonstabilizerness. More precisely, let $S\le\mathbb{F}_2^n$ be a subgroup, and let $R$ be a subset of $\mathbb{F}_2^n$ whose elements lie in pairwise distinct cosets of $S$ (equivalently, $r\oplus r'\notin S$ for any two distinct $r,r'\in R$). Form the coset closure $C=R+S=\{r\oplus s\mid r\in R,\ s\in S\}$, that is, the union of the $|R|$ cosets of $S$ selected by $R$. Then we have $|C|=|R|\,|S|$ according to Lagrange's theorem. The quotient map identifies $R$ with a subset $\bar R$ of $\mathbb{F}_2^n/S$. We now show that the nonstabilizerness of the code state $|C\rangle$ depends only on $\bar R$.

The key observation is how the energies scale under coset closure. We state it as the following lemma.

\begin{lemma}[Energy of a coset union]
\label{lem:scaling}
Let $S\le\mathbb{F}_2^n$ be a subgroup of order $d$, let $a_1,\dots,a_t$ lie in pairwise distinct cosets of $S$, and put
\begin{equation}
 T=\bigsqcup_{i=1}^{t}(a_i\oplus S),\quad
 \bar T=\{\bar a_1,\dots,\bar a_t\}\subseteq\mathbb{F}_2^n/S,
\end{equation}
where $\bar a=a\oplus S$. Then
\begin{equation}
\label{eq:scalinglaw}
 E(T)=d^3E(\bar T),
\end{equation}
with $E(\bar T)=\sum_{\bar y}|\bar T\cap(\bar T\oplus\bar y)|^2$ the additive energy evaluated in the quotient group $\mathbb{F}_2^n/S$.
\end{lemma}

\begin{proof}
Let $A_T(y)=|T\cap(T\oplus y)|$, which equals the number of ordered pairs $(u,v)\in T^2$ with $u\oplus v=y$, so that $E(T)=\sum_yA_T(y)^2$. Fix $y$ and set $\bar y=y\oplus S$. Since the cosets are distinct, every $u\in T$ can be uniquely written as $u=a_i\oplus s$ with $s\in S$, and likewise $v=a_j\oplus s'$. The condition $u\oplus v=y$ then reads
\begin{equation}
 (a_i\oplus a_j)\oplus(s\oplus s')=y .
\end{equation}
If $\bar a_i\oplus\bar a_j\neq\bar y$, then $y\oplus a_i\oplus a_j\notin S$ and there is no solution. If $\bar a_i\oplus\bar a_j=\bar y$, then $y\oplus a_i\oplus a_j\in S$ and $s\oplus s'=y\oplus a_i\oplus a_j$ has exactly $d$ solutions, one for each choice of $s$. Summing over the index pairs,
\begin{equation}
\label{eq:scaling}
 A_T(y)=dA_{\bar T}(\bar y),\quad A_{\bar T}(\bar y)=|\bar T\cap(\bar T\oplus\bar y)| .
\end{equation}
Each coset $\bar y$ contains exactly $d$ elements of $\mathbb{F}_2^n$, so grouping the sum over $y$ by coset yields
\begin{equation}
 E(T)=\sum_{\bar y}\sum_{y\in\bar y}A_T(y)^2
     =d^3\sum_{\bar y}A_{\bar T}(\bar y)^2
     =d^3E(\bar T).
\end{equation}
\end{proof}

Intuitively, the $d^3$ scaling can be read off from the four-tuple count $E(T)=\#\{(a,b,c,d)\in T^4\mid a\oplus b\oplus c\oplus d=0\}$: the quotient equation selects $E(\bar T)$ index quadruples, and for each of them the remaining equation in $S$ leaves three of the four $S$-coordinates free, giving $d^3$.

\begin{theorem}[Coset invariance of nonstabilizerness]
\label{thm:coset}
Let $S\le\mathbb{F}_2^n$ have order $d$, let $R$ be a set of distinct coset representatives with $|R|=2^\ell$, and let $C=R+S$, so $|C|=2^{\ell+\log_2 d}=2^k$. Let $\bar R$ be the projection of $R$ to $\mathbb{F}_2^n/S$. Then
\begin{subequations}
 \begin{align}
 E^{(1)}(R+S)&=d^4 E^{(1)}(\bar R),\\
 M_2(|R+S\rangle)&=M_2(|\bar R\rangle)=4\ell-\log_2 E^{(1)}(\bar R).
 \end{align}
\end{subequations}
Nonstabilizerness depends only on the quotient $C/S$, not on $S$ itself.
\end{theorem}

\begin{proof}
The mechanism is sketched in Fig.~\ref{Fig4}: each coset of $S$ collapses to a single point of the quotient, and $M_2$ only sees the reduced code $\bar R$.

Fix $x$ and set $\bar x=x\oplus S$. We first determine the slice $C\cap(C\oplus x)$, which is the set of $u\in C$ satisfying $u\oplus x\in C$. Writing $u=a_i\oplus s$, we have $u\oplus x\in C$ iff $a_i\oplus x\in a_j\oplus S$ for some $j$, i.e.\ iff $\bar a_i\oplus\bar x\in\bar R$. $a_i\oplus x\in a_j\oplus S$ implies
\begin{equation}
 (a_i\oplus S)\oplus x=(a_j\oplus S),
\end{equation}
so the entire coset $a_i\oplus S$ is contained in the slice. Therefore
\begin{equation}
\label{eq:slice}
 C\cap(C\oplus x)=\bigsqcup_{\bar a\in\bar R\cap(\bar R\oplus\bar x)}(a\oplus S),
\end{equation}
a union of $d$ cosets of $S$ whose coset set is $\bar R\cap(\bar R\oplus\bar x)$. Applying Lemma~\ref{lem:scaling}, we obtain
\begin{equation}
 E\bigl(C\cap(C\oplus x)\bigr)=d^3E\bigl(\bar R\cap(\bar R\oplus\bar x)\bigr).
\end{equation}
Each coset $\bar x$ arises from $d$ values of $x$, so summing over $x$ and grouping by $\bar x$,
\begin{equation}
\begin{aligned}
  E^{(1)}(C)&=\sum_xE\bigl(C\cap(C\oplus x)\bigr)\\
  &=d\sum_{\bar x}d^3E\bigl(\bar R\cap(\bar R\oplus\bar x)\bigr)\\
  &=d^4E^{(1)}(\bar R).
\end{aligned}
\end{equation}
Finally, $|C|=d\,|\bar R|=d\,2^\ell$, hence
\begin{equation}
\begin{aligned}
  M_2(|C\rangle)&=4\log_2|C|-\log_2E^{(1)}(C)\\
  &=4\ell+4\log_2d-\log_2E^{(1)}(\bar R)-4\log_2d\\
  &=4\ell-\log_2E^{(1)}(\bar R),
\end{aligned}
\end{equation}
i.e.\ $M_2(|C\rangle)=M_2(|\bar R\rangle)$.
\end{proof}

A corollary follows immediately. Let $R$ be a Sidon set of fixed size $2^\ell$ and let $S$ range over subgroups of order $d=2^r$, such that the elements of $R$ lie in pairwise distinct cosets of $S$. Then $C=R+S$ is non-stabilizer whenever $\bar R$ is non-affine; it has $k=\ell+r$ logical qubits, and its nonstabilizerness is
\begin{equation}
 M_2(|C\rangle)=M_2(|R\rangle)=3\ell-\log_2(7\cdot2^\ell-6),
\end{equation}
independent of $r$. As $r\to\infty$ the logical dimension grows infinitely, while the nonstabilizerness stays constant, so the nonstabilizerness density $M_2/k$ tends to zero. We can thus obtain non-stabilizer codes with arbitrarily small nonstabilizerness density.

\subsubsection{Exact values for Kerdock codes}
\label{sec:kerdock}

Kerdock codes \cite{Kerdock1972182} are a standard family of nonlinear binary code. For every even $\eta\ge4$ there is a code of length $n=2^\eta$, size $2^{2\eta}$ ($k=2\eta$ logical qubits) and minimum distance $2^{\eta-1}-2^{\eta/2-1}$ \cite{Hammons94TIT, Calderbank97ProcLondonMathSoc, Carlet10book}. The smallest member is the Nordstrom--Robinson code \cite{Nordstrom1967613}, the $(16, 256, 6)$ code at $\eta=4$. The nonstabilizerness can be computed exactly because they are coset closures of Sidon sets, as we shall demonstrate.

A Kerdock code is a union of $2^{\eta-1}$ cosets of the first-order Reed--Muller code \cite{Abbe21TIT} $\mathrm{RM}(1,\eta)$ inside $\mathrm{RM}(2,\eta)$. Writing codewords as evaluations of Boolean functions, each coset is $\{Q_B(x)+\ell(x) \mid \ell\ \text{affine}\}$ for a quadratic form $Q_B$ whose alternating part is a skew-symmetric (zero-diagonal) $\eta\times\eta$ binary matrix $B$. The cosets are indexed by a Kerdock set $K_\eta$ (a set of $2^{\eta-1}$ such matrices) containing the zero matrix, with the property that the difference of any two distinct elements has full rank $\eta$ \cite{Calderbank97ProcLondonMathSoc}. In summary, we have
\begin{equation}
 K(\eta)=\bigcup_{B\in K_\eta}\bigl(Q_B+\mathrm{RM}(1,\eta)\bigr),
\end{equation}
which is exactly the coset closure of Sec.~\ref{sec:coset}, with $S=\mathrm{RM}(1,\eta)$ of order $2^{\eta+1}$ and $R=\{Q_B \mid B\in K_\eta\}$. The quotient $\bar R$ is the Kerdock set $K_\eta$ itself, regarded as a subset of the space $\mathbb{F}_2^{\eta(\eta-1)/2}$ of alternating forms. Hence $\ell=\eta-1$ and, by Theorem~\ref{thm:coset},
\begin{equation}
 M_2\bigl(K(\eta)\bigr)=M_2\bigl(|K_\eta\rangle\bigr),
\label{eq:kerdock-coset}
\end{equation}
independent of the physical length. 

It remains to determine the difference structure of $K_\eta$. We emphasize that the Kerdock set $K_\eta$ of skew-symmetric matrices is not the same object as the set of symmetric matrices that appears in the $\mathbb{Z}_4$/Gray-map construction \cite{Hammons94TIT}; that set is closed under addition (a subgroup) and gives $M_2=0$. Throughout, ``Kerdock set'' means the skew-symmetric set with pairwise full-rank differences.

\begin{lemma}[Kerdock sets are Sidon]
\label{lem:kerdocksidon}
For every even $\eta$, the standard Kerdock set $K_\eta$ is a Sidon set. The map $(a,b)\mapsto B_a+B_b$ is injective on unordered pairs $a\ne b$.
\end{lemma}

\begin{proof}
Let $\eta$ be even and write $U=\mathbb{F}_2^{\eta-1}$ with trace $\Tr:U\to\mathbb{F}_2$. Set $V=\mathbb{F}_2\oplus U\cong\mathbb{F}_2^{\eta}$ with the non-degenerate bilinear form $Q\bigl((\alpha,x),(\beta,y)\bigr)=\alpha\beta+\Tr(xy)$. For $a\in U$ define the $\mathbb{F}_2$-linear map $D(a):V\to V$ by
\begin{equation}
 D(a)(\alpha,x)=\bigl(\Tr(xa),\ \alpha a+a^2x+\Tr(xa)a\bigr),
\end{equation}
and the alternating form $B_a(v,w):=Q(v,D(a)w)$, i.e.
\begin{equation}
\begin{aligned}
 B_a\bigl((\alpha,x),(\beta,y)\bigr)={}&\alpha\Tr(ay)+\beta\Tr(ax)\\
 &+\Tr\bigl((ax)(ay)\bigr)+\Tr(ax)\Tr(ay).
\end{aligned}
\end{equation}
Then $K_\eta=\{B_a \mid a\in U\}$ is the standard Kerdock set, a set of $2^{\eta-1}$ alternating forms with $B_0=0$ and $B_a+B_b$ non-degenerate for $a\ne b$ \cite{Dempwolff15JAlgebraicCombin, Calderbank97ProcLondonMathSoc}.

It suffices to argue on the maps $D(a)$. Since $Q$ is non-degenerate, the correspondence $B\mapsto D$ given by $B(v,w)=Q(v,Dw)$ is $\mathbb{F}_2$-linear and injective, so $B_a+B_b=B_c+B_d$ if and only if $D(a)+D(b)=D(c)+D(d)$. Assume the latter with $a\ne b$, $c\ne d$. Comparing first components gives $\Tr(x(a+b))=\Tr(x(c+d))$ for all $x$, whence by non-degeneracy of the trace form $a+b=c+d=:s$. Writing $b=s+a$, $d=s+c$ and comparing second components, the $\alpha$-terms and the quadratic terms cancel (as $a^2+b^2=(a+b)^2=s^2=c^2+d^2$), leaving
\begin{equation}
 \Tr(xa)a+\Tr\bigl(x(s+a)\bigr)(s+a)=\Tr(xc)c+\Tr\bigl(x(s+c)\bigr)(s+c),
\end{equation}
which expands to $\Tr(xs)a+\Tr(xa)s=\Tr(xs)c+\Tr(xc)s$ for all $x$, or, with $e:=a+c$,
\begin{equation}
 \Tr(xs)e+\Tr(xe)s=0,\quad\forall x\in U .
\end{equation}
If $e=0$ then $a=c$, $b=d$. If $e\ne0$, then since $s=a+b\ne0$ both $x\mapsto\Tr(xs)$ and $x\mapsto\Tr(xe)$ are nonzero linear functionals on $U$. If $s\ne e$ they are distinct, hence have distinct kernels, and some $x$ has $\Tr(xs)=1$, $\Tr(xe)=0$, forcing $e=0$, a contradiction. Therefore $s=e$, i.e.\ $a+c=a+b$, so $c=b$ and $d=s+c=a+b+b=a$. In every case $\{a,b\}=\{c,d\}$.
\end{proof}

The Sidon property of Kerdock sets is easy to check in small $\eta$ cases. For $\eta=4$ the alternating forms on $\mathbb{F}_2^4$ form a space of dimension $\binom 4 2=6$, so there are $64$ forms in all, of which $28$ are nondegenerate. The Kerdock set $K_4$ has $2^{\eta-1}=8$ elements and hence $\binom 8 2=28$ pairwise differences. Each difference is nondegenerate by construction, and a direct computation shows that the $28$ differences are distinct. They are thus the $28$ nondegenerate forms. This verifies the Lemma for $\eta=4$.

The Kerdock code lies between the two extremes of nonstabilizerness. The Kerdock set $K_\eta$ is a Sidon set of size $2^{\eta-1}$, and as a code on $\ell=\eta-1$ qubits it attains the Sidon maximum $M_2=3\ell-\log_2(7\cdot2^\ell-6)\approx2\ell$. The Kerdock code $K(\eta)$ is the coset closure of $K_\eta$ by $\mathrm{RM}(1,\eta)$. By Theorem~\ref{thm:coset} it has the same nonstabilizerness but $k=2\eta$ logical qubits, so $M_2(K(\eta))\approx 2\eta \approx k$. The universal bound of Theorem~\ref{thm:ext} is $M_2<2k$, so the Kerdock code keeps about half of the maximal value. The density with respect to the physical length, $M_2/n$, however, vanishes because the block length $n=2^\eta$ grows exponentially.

\begin{theorem}[Kerdock nonstabilizerness]
\label{thm:kerdock}
For every even $\eta$, the standard Kerdock code $K(\eta)$ has
\begin{equation}
\begin{aligned}
 M_2\bigl(K(\eta)\bigr)&=3(\eta-1)-\log_2\bigl(7\cdot 2^{\eta-1}-6\bigr)\\
 &=2\eta-\log_2 7-2+O(2^{-\eta}).
\end{aligned}
\end{equation}
\end{theorem}

Theorem~\ref{thm:kerdock} follows directly from the Sidon property. Eq.~\eqref{eq:kerdock-coset} and Eq.~\eqref{eq:sidonclosed} (with $2^\ell=2^{\eta-1}$) give the stated value. In particular, the Nordstrom--Robinson code ($\eta=4$) has $M_2=9-\log_2 50\approx3.356$. 

The same coset-closure argument applies to the Delsarte--Goethals codes $\mathrm{DG}(\eta,r)$, which are unions of cosets of $\mathrm{RM}(1,\eta)$ with pairwise coset differences of rank $\ge \eta-2(r-1)$ \cite{Hammons94TIT, Calderbank97ProcLondonMathSoc}. For $r=\eta/2$ one recovers the Kerdock code and the closed form above. For $r<\eta/2$ the differences need not be distinct and the energy must be computed directly. The Preparata code, the $\mathbb{Z}_4$-dual of the Kerdock code \cite{Hammons94TIT}, is not a union of $\mathrm{RM}(1,\eta)$ cosets, so coset closure does not apply directly. Whether all general Kerdock sets are Sidon, including the Kantor-type non-desarguesian constructions of \cite{Dempwolff15JAlgebraicCombin}, reduces to a regularity statement on non-Desarguesian symplectic spreads.

\begin{figure*}
 \centering
 \includegraphics[width=12.9cm]{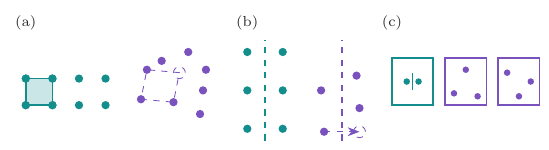}
 \caption{Geometry underlying the gate–nonstabilizerness tradeoff. The codewords of the code are drawn as a point set. (a) Its nonstabilizerness is the irregularity of this set. An affine (stabilizer) code is a perfect lattice. In contrast, a Sidon code is a scatter in which no parallelogram closes (the dashed fourth corner is absent). (b) A diagonal gate is logical if and only if $C$ is invariant under $x\mapsto x\oplus e_i$. The code on the left is symmetric and survives the flip; in the code on the right the flipped copy of one codeword (dashed) lies outside $C$. (c) Symmetry and nonstabilizerness use the same $k$ logical qubits. If $t$ of them are used for symmetry, then $M_2\le 2(k-t)$, i.e. $t\le k-M_2/2$. }
 \label{Fig5}
\end{figure*}

\subsection{Nonstabilizerness and Transversal Gates}
\label{sec:gates}

Universal quantum computation requires non-Clifford gates, and transversal
gates, which act independently on each physical qubit and are fault-tolerant by construction, are the preferred way to realize them \cite{Gottesman98arXiv, Jacoby25PRXQ, Daguerre25PRXQ, Rodriguez25N, Gupta24N, Lee26Feb, Nakhl25May, Chitambar19RMP}. The Eastin--Knill theorem rules out a universal set of transversal gates for any code \cite{Eastin09Mar}. For stabilizer codes the restriction is even stronger: the Bravyi--K\"onig theorem confines their transversal logical gates to a finite level of the Clifford hierarchy \cite{Gottesman99N, Bravyi13Apr}, so only a bounded amount of non-Clifford power is available transversally \cite{Bravyi13Apr, Kobayashi26Jun, Wills25NPhys}. Non-stabilizer codes are not subject to this bound, making them the natural candidates to fill this gap. 

Nonstabilizerness depends only on the quotient and is read off from the fourth moment of the full difference spectrum $A(x)$. A diagonal transversal gate, by contrast, is fixed by the period subgroup 
\begin{equation}
 P(C):=\bigl\{x\in\mathbb{F}_2^n \mid C\oplus x=C\bigr\}=\bigl\{x \mid A(x)=m\bigr\} \le \mathbb{F}_2^n
\label{eq:period}
\end{equation}
of directions that leave the code invariant---the subgroup by which one quotients in Sec.~\ref{sec:coset}. A gate non-Clifford on $t$ coordinates therefore forces $t$ independent directions into $P(C)$, committing at least $t$ of the $k$ logical qubits to invariance. Sec.~\ref{sec:coset} has shown that such commitment is free of nonstabilizerness, and Sec.~\ref{sec:quant-cost} capped each remaining logical qubit at two units, leaving at most $2(k - t)$. Nonstabilizerness and transversal non-Clifford power thus compete for the same additive structure, and a code cannot maximize both.

\subsubsection{Diagonal transversal gates}

Consider a diagonal transversal gate $D(\varphi)=\bigotimes_{i=1}^n\mathrm{diag}(1, \mathrm{e}^{\mathrm{i}\varphi_i})$. Evaluated in the graph-state basis $\{|c_w\rangle=Z^w|G\rangle\}_{w\in\mathbb{F}_2^n}$, which is orthonormal since $\langle G|Z^t|G\rangle=\delta_{t,0}$ \cite{Hein04PRA}, its matrix elements are graph-independent:
\begin{equation}\label{eq:Dvarphi}
\begin{aligned}
 \langle c_{w'}|D(\varphi)|c_w\rangle&=\langle G|Z^{w'\oplus w}D(\varphi)|G\rangle\\
 &=2^{-n}\prod_{i=1}^{n}\Bigl[1+\mathrm{e}^{\mathrm{i}\varphi_i}(-1)^{(w\oplus w')_i}\Bigr]\\
 &=:2^{-n}f\bigl(\varphi; w\oplus w'\bigr).
\end{aligned}
\end{equation}
$D(\varphi)$ is a logical gate when it preserves the code space, i.e., when $\langle c_{w'}|D(\varphi)|c_w\rangle=0$ for all $w\in C$ and $w'\notin C$. The difference $d:=w\oplus w'$ then ranges over $\mathbb{F}_2^n\setminus P(C)$, and the condition is equivalent to $f(\varphi;d)=0$ for $d\notin P(C)$. 

In Eq.~\eqref{eq:Dvarphi}, the $i$-th factor $1+\mathrm{e}^{\mathrm{i}\varphi_i}(-1)^{(w\oplus w')_i}\ne 0$ if and only if the coordinate of $d$ is compatible with $\varphi_i$, i.e., $f(\varphi;d)\ne0$ on the affine subspace
\begin{equation}
 B(\varphi):=\Bigl\{d\in\mathbb{F}_2^n \mid d_i=0\ \text{if}\ \varphi_i=0,\;
 d_i=1\ \text{if}\ \varphi_i=\pi\Bigr\}.
\label{eq:B}
\end{equation}
Note that there is no constraint on $d_i$ if $\varphi_i\ne 0,\pi$ (non-Pauli gate on qubit $i$). Let $F:=\{i \mid \varphi_i\ne 0,\pi\}$ be the set of non-Pauli coordinates and let $z:=\sum_{i:\varphi_i=\pi}e_i$ be the Pauli-$Z$ support, where $e_i\in\mathbb{F}_2^n$ is the $n$-dimensional base vector with the $i$-th component being $1$ and $0$ otherwise. Then $B(\varphi)=z+\operatorname{span}\{e_i \mid i\in F\}$. 

Combining the conditions $f(\varphi;d)=0, d\notin P(C)$ and $f(\varphi;d)\ne0, d\in B(\varphi)$, we conclude that $B(\varphi)\subseteq P(C)$. 

\begin{lemma}[Logical criteria of diagonal transversal gates]
\label{lem:diag}
$D(\varphi)=\bigotimes_{i=1}^n\mathrm{diag}(1, \mathrm{e}^{\mathrm{i}\varphi_i})$ is a logical gate on the CWS code $(G,C)$ if and only if $B(\varphi)\subseteq P(C)$. Equivalently, we have $e_i\in P(C) \;\forall\,i\in F$ and $z\in P(C)$, since $B(\varphi)$ is an affine subspace and $P(C)$ is a linear subspace
\end{lemma}

Lemma~\ref{lem:diag} decides when $D(\varphi)$ descends to a logical unitary. Specifically, the set of logical diagonal transversal gates depends on $C$ alone, not on the graph $G$. The next lemma computes the logical unitary given by $D(\varphi)$ and settles whether it is a Clifford logical gate or not.

\begin{lemma}[Clifford criteria of diagonal transversal gate]
\label{lem:logicalaction}
When $D(\varphi)$ is a logical diagonal transversal gate on $(G,C)$, its restriction to the code space $V_C$ acts on the codeword basis as
\begin{equation}
 U_L:=D(\varphi)\big|_{V_C}
 =\mathrm{e}^{\mathrm{i}\Theta}\,
 \mathcal T_{z}\,
 \prod_{i\in F}\mathcal R_x^{(i)}(\varphi_i),
 \quad
 \Theta=\tfrac12\sum_{i\in F}\varphi_i,
\label{eq:logicalaction}
\end{equation}
where $\mathcal T_z:|c_w\rangle\mapsto|c_{w\oplus z}\rangle$ is the logical translation induced by the physical Pauli $Z^{z}$, and for each $i\in F$
\begin{equation}
 \mathcal R_x^{(i)}(\varphi_i):=\cos\tfrac{\varphi_i}{2}I-\mathrm{i}\sin\tfrac{\varphi_i}{2}\mathcal X^{(i)},
 \quad
 \mathcal X^{(i)}:|c_w\rangle\mapsto|c_{w\oplus e_i}\rangle
\label{eq:logicalrotation}
\end{equation}
is a logical $X$-rotation by $\varphi_i$ on the logical qubit of direction $e_i$. $\mathcal X^{(i)}$ are $|F|$ pairwise commuting independent involutions, so the rotations act on $|F|$ distinct logical qubits and commute with $\mathcal T_z$. Consequently $U_L$ is a product of single-qubit logical gates; it is a logical Clifford gate if and only if $\varphi_i\in\frac\pi2\mathbb Z$ for every $i\in F$.
\end{lemma}

\begin{proof}
For $d=w\oplus w'$ the matrix element computed above reads
\[
 \langle c_{w'}|D(\varphi)|c_w\rangle
 =2^{-n}\prod_{i=1}^{n}\bigl[1+\mathrm{e}^{\mathrm{i}\varphi_i}(-1)^{d_i}\bigr],
\]
and by Lemma~\ref{lem:diag} it vanishes unless $d\in B(\varphi)\subseteq P(C)$. For $d\in B(\varphi)$ each Pauli coordinate (we have $n-|F|$ such coordinates) contributes a factor of $2$, so
\begin{equation}
 \langle c_{w'}|D(\varphi)|c_w\rangle
 =2^{-|F|}\prod_{i\in F}\bigl[1+\mathrm{e}^{\mathrm{i}\varphi_i}(-1)^{d_i}\bigr].
\end{equation}
For $i\in F$,
\begin{equation}
 1+\mathrm{e}^{\mathrm{i}\varphi_i}(-1)^{d_i}
 =\mathrm{e}^{\mathrm{i}\varphi_i/2}\cdot
 \begin{cases}
 2\cos\frac{\varphi_i}{2}, & d_i=0,\\
 -2\mathrm{i}\sin\frac{\varphi_i}{2}, & d_i=1,
 \end{cases}
\end{equation}
which is, up to the prefactor $\mathrm{e}^{\mathrm{i}\varphi_i/2}$, $2$ times the corresponding matrix element of $\mathcal R_x^{(i)}(\varphi_i)=\cos\frac{\varphi_i}{2}I-\mathrm{i}\sin\frac{\varphi_i}{2}\mathcal X^{(i)}$ in the pair $\{|c_w\rangle,|c_{w\oplus e_i}\rangle\}$. Collecting the prefactors into $\Theta=\frac12\sum_{i\in F}\varphi_i$ and using $w'=w\oplus z\oplus\bigoplus_{i\in F}d_ie_i$ yields Eq.~\eqref{eq:logicalaction}. Finally, $\mathcal R_x^{(i)}(\varphi_i)$ is a single-qubit logical Clifford gate for $\varphi_i\in\frac\pi2\mathbb Z$, and $\mathcal T_z$ is a logical Pauli operator, so $U_L$ is Clifford if and only if the action on every coordinate is Clifford.
\end{proof}

A diagonal gate that is non-Pauli on the single coordinate $i$ is logical if and only if $e_i\in P(C)$, that is, if and only if $C$ is invariant under the bit flip $x\mapsto x\oplus e_i$; by Lemma~\ref{lem:logicalaction} such a gate is non-Clifford when $\varphi_i\notin\frac\pi2\mathbb Z$. In particular a transversal $T$ gate on qubit $i$ is logical under this invariance. If $C$ is a Sidon set with $k\ge2$, then $P(C)=\{0\}$ and the only diagonal transversal gate is the identity; maximal nonstabilizerness and a trivial diagonal gate group emerge simultaneously. 

As an example, here we discuss the Kerdock code $K(\eta)=\bigcup_{B\in K_\eta}(Q_B+\mathrm{RM}(1,\eta))$. We first compute the period subgroup. Let $p$ be a period, that is, $p+K(\eta)=K(\eta)$. Since $0\in K(\eta)$, we have $p=p+0\in K(\eta)$, so $p$ is itself a codeword and can be written as $p=Q_{B_0}+\ell$ with $B_0\in K_\eta$ and $\ell$ affine. Adding $Q_{B_0}+\ell$ to the coset $Q_B+\mathrm{RM}(1,\eta)$ gives the coset $Q_{B+B_0}+\mathrm{RM}(1,\eta)$, where we used that $Q_B+Q_{B_0}$ and $Q_{B+B_0}$ differ by an affine function (they have the same alternating part). Thus $p$ sends the coset labeled by $B$ to the coset labeled by $B+B_0$, and it is a period when $K_\eta+B_0=K_\eta$, i.e., when $B_0$ is a period of the Kerdock set $K_\eta$. We now show that $K_\eta$ has no nonzero period. Suppose $K_\eta+B_0=K_\eta$ with $B_0\neq0$. Since $0\in K_\eta$, we get $B_0=0+B_0\in K_\eta$. Choose any $x\in K_\eta\setminus\{0,B_0\}$ (such an $x$ exists because $|K_\eta|=2^{\eta-1}\ge4$), so that $x+B_0\in K_\eta$ as well. Then $\{0,B_0\}$ and $\{x,x+B_0\}$ are two different unordered pairs of elements of $K_\eta$ with the same difference $B_0$, contradicting the Sidon property of $K_\eta$ (Lemma~\ref{lem:kerdocksidon}). Hence $B_0=0$, and the periods of $K(\eta)$ are exactly the elements of $\mathrm{RM}(1,\eta)$, each of which fixes every coset:
\begin{equation}
  P\bigl(K(\eta)\bigr)=\mathrm{RM}(1,\eta), \quad \bigl|P\bigl(K(\eta)\bigr)\bigr|=2^{\eta+1}.
\end{equation}

By Lemma~\ref{lem:diag}, a diagonal transversal gate $D(\varphi)$ of $K(\eta)$ is logical if and only if $B(\varphi)\subseteq P(K(\eta))=\mathrm{RM}(1,\eta)$. Suppose some coordinate $i$ has a non-Pauli phase ($\varphi_i\neq0,\pi$). Then $e_i\in B(\varphi)$ and hence $e_i\in\mathrm{RM}(1,\eta)$. But an affine function on $\mathbb{F}_2^{\eta}$ has weight $0$, $2^{\eta-1}$ or $2^\eta$, and a weight-one vector cannot be affine. So every $\varphi_i$ is $0$ or $\pi$, and $D(\varphi)=Z^z$ with $z\in\mathrm{RM}(1,\eta)$. The diagonal transversal gates of $K(\eta)$ are the Pauli gates $Z^z$, $z\in\mathrm{RM}(1,\eta)$, which include $Z^{\otimes n}$ (given by the constant affine function) and the hyperplane gates $Z^{1_H}$ with $H=\{x\mid\ell(x)=1\}$ for nonzero linear $\ell$. Each of them acts on the codeword basis as the translation $w\mapsto w\oplus z$, and is nontrivial when $z\neq0$. In particular, the Kerdock and Nordstrom--Robinson code states have no non-Pauli diagonal transversal gate, and hence no non-Clifford one. Their diagonal transversal gates form the Pauli subgroup $\{Z^z\}_{z\in\mathrm{RM}(1,\eta)}$, which is determined by the period subgroup $P(C)$.

A similar computation gives the transversal Pauli group of a general CWS code. Since
\begin{equation}
  X^aZ^b|c_w\rangle=(-1)^{a\cdot(b\oplus w)}|c_{w\oplus b\oplus\Gamma a}\rangle,
\end{equation}
the operator $X^aZ^b$ is a logical gate if and only if $b\oplus\Gamma a\in P(C)$ (recall that $\Gamma$ is the adjacent matrix of the graph). Hence the transversal Pauli group is
\begin{equation}
 \bigl\{X^aZ^b \mid b\oplus\Gamma a\in P(C)\bigr\},
\end{equation}
which is determined by $P(C)$ and $\Gamma$. The single-qubit diagonal operators $Z_i$ (that is, $a=0$, $b=e_i$) are logical when $e_i\in P(C)$. This means that $C$ contains two codewords differing only in coordinate $i$, so that $Z_i$ is a weight-one logical operator and the code distance is $1$. A code of distance $\ge2$ has $e_i\notin P(C)$ for every $i$, and by Lemma~\ref{lem:diag} it has no nontrivial single-qubit diagonal transversal gate. Multi-qubit diagonal gates are still possible: $Z^z$ is logical whenever $z\in P(C)$, as the hyperplane gates of the Kerdock code show.

The tradeoff between nonstabilizerness and transversal non-Clifford power \cite{Bravyi13Apr} is then a direct application of the coset and extremal theorems. An intuitive picture is shown in Fig.~\ref{Fig5}.

\begin{theorem}[Gate--nonstabilizerness tradeoff]
\label{thm:tradeoff}
If the CWS code $(G,C)$ admits a logical diagonal transversal gate $D(\varphi)$ that is non-Clifford on $t$ coordinates, then
\begin{equation}
 M_2(|C\rangle) \le 2 (k-t).
\end{equation}
Equality holds only in the affine case.
\end{theorem}

\begin{proof}
A non-Clifford coordinate has $\varphi_i\notin\frac\pi2\mathbb Z$, in particular $\varphi_i\ne 0,\pi$. By Lemma~\ref{lem:diag}, its coordinate vector lies in $P(C)$; hence $\langle e_{i_1},\dots,e_{i_t}\rangle\le P(C)$ and $|P(C)|\ge2^t$. Writing $C=R+P(C)$ with $|R|=2^{k-r}$ and $r=\log_2 |P(C)|\ge t$, Theorem~\ref{thm:coset} gives $M_2(|C\rangle)=M_2(\bar R)$, and Theorem~\ref{thm:ext} gives $M_2(\bar R)<2(k-r)\le2(k-t)$.
\end{proof}

The extension $M_2^{\max}(C)\le 2(k-t)$ is false. For instance, for the affine code $C=\mathbb{F}_2^k$ on the empty graph, every coordinate vector lies in $P(C)=\mathbb{F}_2^k$, so $T^{\otimes k}$ is a logical transversal gate non-Clifford on all $t=k$ coordinates; the right-hand side is $2(k-t)=0$, yet by Proposition~\ref{prop:affinecap} the capacity equals $\mu_k\approx k>0$ (note that an affine code space is Clifford-equivalent to the full $k$-qubit Hilbert space).
The $P(C)$-orbit of any codeword stays inside the code space and spans an affine copy of $\mathbb C^{|P(C)|}$, on which the nonstabilizerness can reach $\mu_r$. Hence we have $M_2^{\max}(C)\ \ge\ \mu_r\ \ge\ \mu_t$---a transversal non-Clifford gate gives a lower bound on the capacity. 

\subsubsection{Beyond diagonal gates}

We now discuss the general single-qubit transversal gates. As $M_2(|C\rangle)>0$ is equivalent to $C$ non-affine, one might hope that it certifies the existence of a transversal non-Clifford gate. However, this is not the case; the correct criterion is an axis condition on $P(C)$.

Let $V_i=\alpha I+\beta X+\gamma Y+\delta Z$ be a single-qubit unitary on qubit $i$. On a non-isolated qubit (i.e., $\Gamma e_i\ne0$), the four displacements $0$, $e_i$, $\Gamma e_i$, and $e_i+\Gamma e_i$ are pairwise distinct, and $V_i$ is a logical gate if and only if
\begin{subequations}\label{eq:single-criterion}
 \begin{align}
 \delta\ne0\ &\Rightarrow\ e_i\in P(C),\\
 \beta\ne0\ &\Rightarrow\ \Gamma e_i\in P(C),\\
 \gamma\ne0\ &\Rightarrow\ e_i+\Gamma e_i\in P(C).
 \end{align}
\end{subequations}
The coefficient $\alpha$ is unconstrained. This follows by expanding $V_i$ in the Pauli basis and acting on a codeword. Each Pauli component displaces the codeword by one of the four vectors, and the gate is logical when every nonzero component displaces codewords within the code. For the axis rotations we have
\begin{subequations}
 \begin{align}
 R_z(\theta)_i\ \text{logical}&\iff e_i\in P(C),\\
 R_x(\theta)_i\ \text{logical}&\iff \Gamma e_i\in P(C),\\
 R_y(\theta)_i\ \text{logical}&\iff e_i+\Gamma e_i\in P(C)
 \end{align}
\end{subequations}
for $\theta\ne0\pmod{2\pi}$, and each is a non-Clifford logical gate when $\theta\notin\frac\pi2\mathbb Z$ (the $R_z$ case is the diagonal logical criterion of Lemma~\ref{lem:diag}). Consequently, there is a transversal single-qubit non-Clifford gate on qubit $i$ if and only if $\{e_i,\Gamma e_i,e_i+\Gamma e_i\}\cap P(C)\ne\varnothing$. 

For general multi-qubit product-gate, distinct displacements may coalesce and phases may cancel. Expand $V=\bigotimes_i V_i=\sum_{a,b}\mu(a,b)X^aZ^b$ with factorized $\mu(a,b)=\prod_i\mu_i(a_i,b_i)$, where $\mu_i(0,0)=\alpha_i$, $\mu_i(1,0)=\beta_i$, $\mu_i(0,1)=\delta_i$, $\mu_i(1,1)=\mathrm{i}\gamma_i$. Acting on a codeword, $X^aZ^bZ^w|G\rangle=(-1)^{a\cdot(b+w)}Z^{b+w+\Gamma a}|G\rangle$ and grouping by the displacement $d=b+\Gamma a$ gives
\begin{equation}
 V|c_w\rangle=\sum_d\Bigl[\sum_a\mu(a,d+\Gamma a)(-1)^{a\cdot(d+w)}\Bigr]|c_{w+d}\rangle .
\end{equation}
The gate $V$ is logical if and only if, for every $d\notin P(C)$ and every $w\in C$ with $w+d\notin C$,
\begin{equation}
 \sum_a\mu(a,d+\Gamma a)(-1)^{a\cdot(d+w)}=0 .
\end{equation}
For a single non-isolated qubit each displacement $d$ is attained by a unique $a$, the sum is a single term, and this reduces to Eq.~\eqref{eq:single-criterion}. For general multi-qubit gates distinct displacements may coalesce and the phase cancellation may render $V$ logical without any single coordinate forcing an axis vector into $P(C)$. Characterizing which configurations of $\{a \mid \mu(a,d+\Gamma a)\ne0\}$ relative to $C^{\perp}$ realize this condition, however, remains open.

The case of isolated vertices is remarkable. If $\Gamma e_i=0$, then $e_i+\Gamma e_i=e_i$, the $X$-axis condition holds vacuously. An $X$-rotation $R_x(\theta)_i$ is logical for any $C$, acting as $|c_w\rangle\mapsto \mathrm{e}^{\mathrm{i}(\theta/2)(-1)^{w_i+1}}|c_w\rangle$. On the empty graph every qubit is isolated, so every CWS code, including the Sidon codes of maximal nonstabilizerness whose diagonal gate group is trivial, carries the transversal non-Clifford gates $R_x(\theta)_i$ ($\theta\notin\frac\pi2\mathbb Z$) for all $i$. Thus, maximal nonstabilizerness and a trivial diagonal gate group do not preclude transversal non-Clifford power; the gate is simply non-diagonal. This does not conflict with the Bravyi--K\"onig theorem, as empty-graph codewords are product states and corresponds to distance-$1$ codes that cannot detect errors.

\section{Other non-stabilizer Codes}
\label{sec:other}
\subsection{Cubic-phase Codes and non-Abelian Topological Order}

The phase dictionary applies directly to cohomological code families. A twisted quantum double $D^\omega(G)$ is defined by a three-cocycle $\omega\in H^3(G,\mathrm{U}(1))$. After the untwisted (stabilizer) background is trivialized, its code state is a phase state with cubic $f$, the phase of the cocycle \cite{Kobayashi26Jun,Ellison22PRXQ}. The nonstabilizerness is then a statement about the trilinear form $B$ alone. The minimal instance is the twisted $\mathbb{Z}_2^3$ double, for which $H^3(\mathbb{Z}_2^3,\mathrm{U}(1))=\mathbb{Z}_2$ and the nontrivial cocycle is the triple cup product $\omega=a\cup b\cup c$. This model is equivalent to the $D_4$ topological order \cite{Kobayashi26Jun}, and its decoupled state is the determinant/CCZ state, with nonstabilizerness $M_2=6-\log_2 22$. More generally, non-Abelian topological order carries intrinsic nonstabilizerness. A two-dimensional order is realizable by Pauli stabilizers iff its anyons are Abelian, and non-Abelian string-net ground states have extensive, long-range nonstabilizerness that no finite-depth local unitary can remove \cite{Ellison22PRXQ,Zhang26arXiv}.

As a benchmark family of exactly solvable cubic codes we single out $f=x_0 Q(\mathbf y)$, the product of a linear coordinate $x_0$ with a quadratic form $Q$ on the remaining $n'=n-1$ coordinates whose polar form $B_Q$ has maximal rank, i.e., $\rank B_Q=n'$ for even $n'$ and $\rank B_Q=n'-1$ for odd $n'$ (note that $B$ depends on $Q$ only through $B_Q$). By the Leibniz rule for discrete derivatives, we have
\begin{equation}
  B(u,v,w)=u_0 B_Q(v_y,w_y)+v_0 B_Q(w_y,u_y)+w_0 B_Q(u_y,v_y),
\end{equation}
where $u=(u_0,u_y), v= (v_0, v_y)$, and $w=(w_0, w_y)$. The condition $B(u,v,\cdot)\equiv0$ splits into $B_Q(u_y,v_y)=0$ and $u_0B_Q(v_y,\cdot)+v_0B_Q(\cdot,u_y)\equiv0$, and counting the pairs $(u,v)$ (distinguishing $u_y=0$, $u_y$ in the radical of $B_Q$, and $u_y$ generic) gives
\begin{equation}
\label{eq:x0Q}
\begin{aligned}
  N(B)&=\begin{cases}2^{n'+2}+2^{2n'-1}-2^{n'-1},& n'\ \text{even},\\
  2^{n'+3}+2^{2n'-1}-2^{n'},& n'\ \text{odd},\end{cases}
  \\
  M_2&=2(n'{+}1)-\log_2 N(B).
\end{aligned}
\end{equation}
For $n'=2,4,6,8$ this gives $N=22,184,2272,33664$, i.e., $M_2\approx 1.541,2.476,2.850,2.961$. The family is the phase-side counterpart of the Kerdock construction---a linear coordinate times a symplectic form, mirroring how the Kerdock set is built from skew (symplectic) matrices \cite{Calderbank97ProcLondonMathSoc}. Moreover, it is easy to check that for $n\le5$ it is the exact maximizer of nonstabilizerness.

\begin{table*}
\caption{Summary of code families, their inert backgrounds, and the corresponding nonstabilizerness carriers.}
\begin{tabular}{lll}
\hline
code family & inert background  & nonstabilizerness carrier \\
\hline
permutation-invariant & symmetric (empty/complete) graph & weight profile \\
CWS codes & graph state (quadratic form) & classical code $C$ (additive energy) \\
cubic / cohomological & quadratic phase (Clifford sector) & three-cocycle $\Delta_u\Delta_v\Delta_w f$ (rank data $N(B)$) \\
group algebra ($D(G)$) & cyclic group law (Abelian sector) & non-cyclic / non-Abelian law (cubic defect) \\
\hline
\end{tabular}
\label{tab:Summary}
\end{table*}

\subsection{Non-Abelian Group States}

As a further instance we isolate the group-theoretic core of a quantum double $D(G)$ \cite{Kitaev03AnnPhys}, the \emph{multiplication state}
\begin{equation}
  \ket{\psi_G}=|G|^{-1}\sum_{g,h\in G}\ket g\ket h\ket{gh}
\end{equation}
on three qudits of dimension $|G|$, i.e., the image of $\ket{+}\ket{+}\ket0$ under the ``controlled group multiplication'' gate $U_{\mathrm{mult}}\ket{g,h,0}=\ket{g,h,gh}$. The qudit Pauli frame is fixed by the cyclic shift $X\ket j=\ket{j+1 \mod d}$, $d=|G|$, and all Pauli operators below refer to this cyclic frame. Relative to this frame the state carries nonstabilizerness precisely when the group law is not cyclic, while for a cyclic group the state is manifestly stabilizer. With $gh=(g+h) \mod d$,
$\ket{\psi_G}=d^{-1}\sum_{g,h}\ket{g,h,g+h}$ is fixed by the $d^3$ commuting generalized Paulis
\begin{equation}
  \{X^a\otimes X^b\otimes X^{a+b}\}_{a,b\in\mathbb{Z}_d}\ \cup\ \{(Z\otimes Z\otimes Z^{-1})^t\}_{t\in\mathbb{Z}_d}.
\end{equation}
The first family merely permutes the terms of the uniform superposition and the second has eigenvalue $\omega^{t(g+h-(g+h))}=1$ on the support. The two families commute since $\omega^{at+bt-(a+b)t}=1$. Conversely, if the group law is not cyclic, then for some $a,b,c$ the diagonal expectation $\langle\psi_G|Z^aZ^bZ^c|\psi_G\rangle=d^{-2}\sum_{g,h}\omega^{ag+bh+c(gh)}$ has magnitude strictly between $0$ and $1$, because $(g,h)\mapsto\omega^{c(gh)}$ extends to a character of $G\times G$ only for a cyclic law. Such a Pauli expectation is forbidden for a stabilizer state, so the state is non-stabilizer. Hence $M_2(\ket{\psi_G})=0$ if and only if the group law is cyclic.

The minimal non-cyclic group is the Klein four-group, and it exposes the bridge to the cubic picture. For $G=\mathbb{Z}_2\times\mathbb{Z}_2$ labeled by $0\mapsto e,\ 1\mapsto a,\ 2\mapsto b,\ 3\mapsto ab$, the Klein law is $\mathbb{F}_2^2$ addition. In the $\mathbb{F}_2^2$ (two-qubit) frame the state factors as $\ket{\psi_G}=\ket{\psi_\mathrm{add}}^{\otimes2}$, with $\ket{\psi_\mathrm{add}}=2^{-1}\sum_{a,b}\ket a\ket b\ket{a+b}$ a stabilizer (parity) state. The non-trivial content appears only in the cyclic $\mathbb{Z}_4$ frame, where the defect of $\mathbb{F}_2^2$ addition relative to cyclic addition is the triple cup product. The fourth-moment sum evaluates to $22$, the dependent-pair count of the determinant on $\mathbb{F}_2^3$, so that
\begin{equation}
  M_2(\ket{\psi_{\mathbb{Z}_2\times\mathbb{Z}_2}})=6-\log_222,
\end{equation}
which is exactly the CCZ value. For larger non-Abelian groups the nonstabilizerness is non-vanishing, with a magnitude depending on the labeling. In a fixed natural labeling 
\footnote{Here a group of order $d$ is labeled by its elements $0,\dots,d-1$ in the order given below, relative to the cyclic qudit Pauli frame $X\lvert j\rangle=\lvert j+1\bmod d\rangle$, $Z\lvert j\rangle=\omega^{j}\lvert j\rangle$ with $\omega=\mathrm{e}^{2\pi \mathrm{i}/d}$. For $S_3=\langle r,s\mid r^3=s^2=1,\ srs=r^{-1}\rangle$ the element $r^k s^m$ carries the label $k+3m$; for the dihedral group $D_4=\langle r,s\mid r^4=s^2=1,\ srs=r^{-1}\rangle$ of order $8$ (not the $D_4$ topological-order code of the preceding section) it carries $k+4m$; and for $Q_8=\{\pm1,\pm i,\pm j,\pm k\}$ the order is $1,-1,i,-i,j,-j,k,-k$.}
one finds $M_2(S_3)=4.38$, $M_2(D_4)=4.39$, $M_2(Q_8)=6.13$, while $M_2(\mathbb{Z}_3)=M_2(\mathbb{Z}_4)=M_2(\mathbb{Z}_5)=M_2(\mathbb{Z}_6)=0$. The magnitude does not reduce to a single elementary group invariant ($D_4$ and $Q_8$ share order and number of conjugacy classes yet differ), reflecting that the nonstabilizerness is a genuine non-Abelian additive energy of the multiplication map.

\section{Conclusion}
\label{sec:conc}

In summary, the inert background is Clifford and removable, while the nonstabilizerness is cubic and countable. Table \ref{tab:Summary} summarizes the nonstabilizerness structure of different code families. This converts an exponential Pauli sum into finite classical counting, gives exact formulae for entire code families, identifies their maximally magical members, and provides a unified framework for nonstabilizerness evaluation in fault-tolerant quantum codes. We anticipate that, with nonstabilizerness computable in closed form, magic will take its place alongside distance and rate as an ordinary design parameter of quantum codes.

\section*{Acknowledgments}
This work is supported by the Fundamental and Interdisciplinary Disciplines Breakthrough Plan of the Ministry of Education of China (Grant Nos.~JYB2025XDXM115 and JYB2025XDXM201) and the Beijing Science and Technology Planning Project (Grant No.~Z25110100040000).

\section*{Data Availability}
No data were created or analyzed in this study.

\appendix

\section{Detailed proof of the graph--code decoupling Lemma}
\label{sec:decoupling-proof}

Here we prove Lemma~\ref{lem:decouple} in 5 steps. 

\emph{Step 1}: The graph state in the Hadamard basis.
The $n$-qubit graph state associated to $G=(V,E)$ is
\begin{equation}\label{eq:apx-gs}
 |G\rangle = \prod_{(i,j)\in E} \mathrm{CZ}_{ij}\, |+\rangle^{\otimes n},
\end{equation}
where the product runs over all undirected edges and the order is immaterial since all $\mathrm{CZ}$ gates commute. Using $|+\rangle = H|0\rangle$ and $H^{\otimes n} H^{\otimes n} = I$, we rewrite
\begin{equation}\label{eq:apx-gs2}
 |G\rangle = \Bigl(\prod_{e\in E} \mathrm{CZ}_e\Bigr) H^{\otimes n} |0\rangle^{\otimes n}.
\end{equation}

The encoded state is $|\psi_{\mathrm{enc}}\rangle = \sum_{w\in C} \psi_w Z^w |G\rangle$, with $Z^w \equiv \bigotimes_{i=1}^n Z^{w_i}$. Apply $H^{\otimes n}$:
\begin{equation}\label{eq:apx-Hpsi}
 H^{\otimes n}|\psi_{\mathrm{enc}}\rangle = \sum_{w\in C} \psi_w H^{\otimes n} Z^w |G\rangle.
\end{equation}
Using $H Z = X H$ and tensoring over $n$ qubits, we have
\begin{equation}\label{eq:apx-HZ-XH}
 H^{\otimes n} Z^w = \bigotimes_{i=1}^n H Z^{w_i} = \bigotimes_{i=1}^n X^{w_i} H = X^w H^{\otimes n}.
\end{equation}
Substituting \eqref{eq:apx-HZ-XH} into \eqref{eq:apx-Hpsi} leads to
\begin{equation}\label{eq:apx-Hpsi2}
 H^{\otimes n}|\psi_{\mathrm{enc}}\rangle = \sum_{w\in C} \psi_w X^w H^{\otimes n} |G\rangle.
\end{equation}

\emph{Step 2}: Conjugation of the entangling circuit.
Inserting Eq.~\eqref{eq:apx-gs2} into $H^{\otimes n}|G\rangle$, we have
\begin{equation}\label{eq:apx-conj}
\begin{aligned}
 H^{\otimes n}|G\rangle & = H^{\otimes n} \Bigl(\prod_{e\in E} \mathrm{CZ}_e\Bigr) H^{\otimes n} |0\rangle^{\otimes n}\\
 & = \prod_{e\in E} \bigl(H^{\otimes n} \mathrm{CZ}_e H^{\otimes n}\bigr) |0\rangle^{\otimes n},
\end{aligned}
\end{equation}
where we used $H^{\otimes n} H^{\otimes n} = I$ to insert the identity between successive $\mathrm{CZ}$ gates and convert the conjugated product into a product of conjugates.

For a specific edge $e=(i,j)$, the gate $\mathrm{CZ}_{ij}$ acts as identity on all qubits except $i,j$. The Hadamards on the unaffected qubits commute past $\mathrm{CZ}_{ij}$ and cancel with their counterparts from the right, i.e., for $k\notin\{i,j\}$, $H_k \mathrm{CZ}_{ij} H_k = \mathrm{CZ}_{ij}$. Hence
\begin{equation}\label{eq:apx-edgegate}
 H^{\otimes n} \mathrm{CZ}_{ij} H^{\otimes n} = (H_i \otimes H_j)\mathrm{CZ}_{ij}(H_i \otimes H_j) \otimes I_{\widehat{ij}} \equiv \tilde{U}_{ij},
\end{equation}
where $I_{\widehat{ij}}$ is the identity on all qubits other than $i,j$. We define the decoupling edge gate
\begin{equation}\label{eq:apx-tildeU-edge}
 \tilde{U}_{ij} \equiv (H_i \otimes H_j)\,\mathrm{CZ}_{ij}\,(H_i \otimes H_j),
\end{equation}
which is a two-qubit Clifford unitary. The full decoupling circuit is
\begin{equation}\label{eq:apx-tildeUG}
 \tilde{U}_G \equiv \prod_{(i,j)\in E} \tilde{U}_{ij},
\end{equation}
and \eqref{eq:apx-conj} becomes simply
\begin{equation}\label{eq:apx-HG}
 H^{\otimes n}|G\rangle = \tilde{U}_G |0\rangle^{\otimes n}.
\end{equation}

\emph{Step 3}: The commutation relation $[X^w, \tilde{U}_G]=0$, i.e., for any edge $e=(i,j)$, the single-qubit Pauli operators $X_i$ and $X_j$ each commute with $\tilde{U}_{ij}$.

 Using $X H = H Z$, we compute
 \begin{equation}\label{eq:apx-Xi-tildeU}
 \begin{aligned}
 X_i \tilde{U}_{ij}
 &= X_i (H_i \otimes H_j)\mathrm{CZ}_{ij}(H_i \otimes H_j) \\
 &= (H_i Z_i \otimes H_j)\mathrm{CZ}_{ij}(H_i \otimes H_j) \\
 &= (H_i \otimes H_j)\,(Z_i \otimes I_j)\mathrm{CZ}_{ij}(H_i \otimes H_j).
 \end{aligned}
 \end{equation}
 Both $Z_i\otimes I_j$ and $\mathrm{CZ}_{ij}$ are diagonal in the computational basis, hence they commute:
 \begin{equation}\label{eq:apx-Z-CZ-commute}
 (Z_i \otimes I_j) \mathrm{CZ}_{ij} = \mathrm{CZ}_{ij}(Z_i \otimes I_j).
 \end{equation}
 Inserting this into \eqref{eq:apx-Xi-tildeU} and using $(H_i\otimes H_j)^2 = I$:
 \begin{equation}\label{eq:apx-Xi-final}
 \begin{aligned}
 X_i \,\tilde{U}_{ij}
 &= (H_i \otimes H_j)\mathrm{CZ}_{ij}(Z_i \otimes I_j)(H_i \otimes H_j) \\
 &= (H_i \otimes H_j)\mathrm{CZ}_{ij}(H_i \otimes H_j) (H_i \otimes H_j)\\
 &\quad\times(Z_i \otimes I_j)(H_i \otimes H_j)\\
 &= \tilde{U}_{ij}(H_i Z_i H_i \otimes I_j) \\
 &= \tilde{U}_{ij}X_i,
 \end{aligned}
 \end{equation}
 where the last equality uses $H Z H = X$. By symmetry, $[X_j, \tilde{U}_{ij}] = 0$.

For any qubit $k\notin\{i,j\}$, $X_k$ acts on a different Hilbert space than $\tilde{U}_{ij}$ and trivially commutes. Consequently, for \emph{every} $w\in\mathbb{F}_2^n$, $[X^w, \tilde{U}_{ij}] = 0$, and since $\tilde{U}_G$ is the product of edge gates $\tilde{U}_{ij}$,
\begin{equation}\label{eq:apx-commute-G}
 [X^w, \tilde{U}_G] = 0 \qquad \forall\, w\in\mathbb{F}_2^n.
\end{equation}

\emph{Step 4}: Completion of the main identity.
Insert \eqref{eq:apx-HG} into \eqref{eq:apx-Hpsi2}:
\begin{equation}\label{eq:apx-final}
\begin{aligned}
  H^{\otimes n} |\psi_{\mathrm{enc}}\rangle &= \sum_{w\in C} \psi_w X^w \tilde{U}_G |0\rangle^{\otimes n} = \tilde{U}_G \sum_{w\in C} \psi_w X^w |0\rangle^{\otimes n}\\
  &= \tilde{U}_G \sum_{w\in C} \psi_w |w\rangle,
\end{aligned}
\end{equation}
where we used $X^w|0\rangle^{\otimes n} = \bigotimes_{i=1}^n X^{w_i}|0\rangle
= \bigotimes_{i=1}^n |w_i\rangle = |w\rangle$. Equation~\eqref{eq:apx-final} is precisely the decoupling identity $H^{\otimes n}|\psi_{\mathrm{enc}}\rangle = \tilde{U}_G \sum_{w\in C}\psi_w|w\rangle$, establishing Lemma~\ref{lem:decouple}. 

\emph{Step 5}: Circuit depth via Vizing's theorem.
Each $\tilde{U}_{ij}$ acts on qubits $i,j$. Two such gates on disjoint edges act on disjoint qubit pairs and can be executed in the same layer. Hence the minimum number of parallel layers equals the edge-chromatic number $\chi'(G)$: the minimal number of colors needed to assign each edge a color such that incident edges have different colors.

Vizing's theorem \cite{Diestel2017} states that for any simple graph,
\begin{equation}\label{eq:apx-vizing}
 \Delta(G) \le \chi'(G) \le \Delta(G) + 1,
\end{equation}
where $\Delta(G) = \max_{v\in V} \deg(v)$. The lower bound is trivial (all edges at a degree-$\Delta$ vertex need distinct colors). Hence $D(\tilde{U}_G) = \chi'(G) \le \Delta(G)+1$.

For bounded-degree graphs ($\Delta = O(1)$), the circuit depth is $O(1)$, independent of the code length $n$. For the square lattice ($\Delta=4$), $\chi'=4$ by K\"{o}nig's theorem \cite{Diestel2017} on bipartite graphs; for the honeycomb lattice ($\Delta=3$), $\chi'=3$; for the triangular lattice ($\Delta=6$), $\chi'=6$ or $7$, all constant.

\section{Higher R\'enyi orders}
\label{sec:appalpha}

We record the order-$\alpha$ generalization of the dictionary and its coset invariance. For $S\subseteq\mathbb{F}_2^n$ and integer $\alpha\ge2$ define the order-$2\alpha$ energy
\begin{equation}
 E_\alpha(S):=\#\bigl\{(w,w')\in S^{\alpha}\times S^{\alpha}\mid\
 w_1\oplus\cdots\oplus w_\alpha=w'_1\oplus\cdots\oplus w'_\alpha\bigr\},
\end{equation}
so that $E_2=E$, and the $\alpha$-fold slice energy
\begin{equation}
 E^{(\alpha-1)}(C):=\sum_{x}E_\alpha\bigl(C\cap(C\oplus x)\bigr),
\label{eq:Ealpha}
\end{equation}
reducing to Eq.~\eqref{eq:E1} at $\alpha=2$.

For the subset state $|C\rangle$ with $|C|=m=2^k$,
\begin{equation}
 \sum_{P\in\mathcal P_n}\langle C|P|C\rangle^{2\alpha}=\frac{2^{n}}{m^{2\alpha}}E^{(\alpha-1)}(C),
\end{equation}
\begin{equation}
 M_\alpha(|C\rangle)=\frac{2\alpha k-\log_2 E^{(\alpha-1)}(C)}{\alpha-1},
\end{equation}
where $M_\alpha=\frac{1}{\alpha-1}[n-\log_2\sum_P\langle C|P|C\rangle^{2\alpha}]$. The derivation repeats that of Eq.~\eqref{eq:M2full} with the fourth moment replaced by the $2\alpha$-th: for $S_x=C\cap(C\oplus x)$ one has $\langle C|X^xZ^z|C\rangle=\frac1m\widehat S_x(z)$ and the $2\alpha$-th-moment form of Parseval, $\sum_z|\widehat S_x(z)|^{2\alpha}=2^nE_\alpha(S_x)$, obtained by expanding $|\widehat S_x(z)|^{2\alpha}$ and projecting onto $\bigoplus_i w_i\oplus\bigoplus_i w'_i=0$. If $C$ is affine, then $S_x$ is $C$ or empty, $E_\alpha(C)=m^{2\alpha-1}$, and $E^{(\alpha-1)}(C)=m^{2\alpha}$, so $M_\alpha=0$ for all $\alpha$. The weighted dictionary of Lemma~\ref{lem:dict} lifts to every order by the same derivation. With $f_x(w)=\overline{\psi_{w\oplus x}}\psi_w$, define the weighted order-$2\alpha$ energy
\begin{equation}
 \mathcal E^{\psi}_\alpha(S_x):=\sum_{\substack{w^{(1)},\dots,w^{(\alpha)}\in S_x\\w'^{(1)},\dots,w'^{(\alpha)}\in S_x\\ \bigoplus_i w^{(i)}=\bigoplus_j w'^{(j)}}}
 \prod_{i=1}^{\alpha}f_x(w^{(i)})\prod_{j=1}^{\alpha}\overline{f_x(w'^{(j)})},
\end{equation}
and $\mathcal E^{(\alpha-1)}_{\psi}(C):=\sum_x\mathcal E^{\psi}_\alpha(S_x)$, so that $\sum_{P\in\mathcal P_n}\langle\psi|P|\psi\rangle^{2\alpha}=2^n\mathcal E^{(\alpha-1)}_{\psi}(C)$ and
\begin{equation}
 M_\alpha(|\psi\rangle)=-\frac{1}{\alpha-1}\log_2\mathcal E^{(\alpha-1)}_{\psi}(C),
\end{equation}
reducing to $\mathcal E^{(\alpha-1)}_{\psi}(C)=m^{-2\alpha}E^{(\alpha-1)}(C)$ for the uniform state.

The coset-invariance theorem also lifts verbatim. A union $T=\bigsqcup_{\bar t\in T}(t+S)$ of cosets of a subgroup $S$ of order $d$ has $E_\alpha(S\cdot T)=d^{2\alpha-1}E_\alpha(T)$: in the quotient the condition $\bigoplus_i\bar w_i=\bigoplus_i\bar w'_i$ selects the cosets, and for each quotient solution the subgroup equation leaves $2\alpha-1$ of the $2\alpha$ subgroup coordinates free. Applying this to $(R+S)\cap((R+S)\oplus x)$, which is a coset union with coset set $\bar R\cap(\bar R\oplus\bar x)$, and summing over the $d$ elements $x$ per coset,
\begin{equation}
 E^{(\alpha-1)}(R+S)=d^{2\alpha} E^{(\alpha-1)}(\bar R),\;
 M_\alpha(|R+S\rangle)=M_\alpha(|\bar R\rangle),
\end{equation}
the factor $2\alpha\log_2 d$ cancelling against $\log_2 d^{2\alpha}$. Thus $M_\alpha$ is a coset-level invariant for every integer $\alpha\ge2$. We note that the naively nested sum $\sum_{x_1,\dots,x_{\alpha-1}}E(C\cap\bigcap_i(C\oplus x_i))$ is a different object for $\alpha\ge3$, scaling as $d^{\alpha+2}$ rather than $d^{2\alpha}$ under coset closure, and it does not enter the $M_\alpha$ dictionary; the two objects coincide only at $\alpha=2$.

\begin{table}[ht]
\caption{Extremal data in the half-space regime $n=k+1$, from the knapsack and congruence above (exact solution by dynamic programming).}
\centering
\begin{tabular}{lccccc}
\toprule
$k$ & $n$ & $N_{\min}$ & $E^{(1)}$ & $M_2$ & status\\
\midrule
$3$ & $4$ & $5$ & $1240$ & $1.7239$ & exact (exhaustive)\\
$4$ & $5$ & $59$ & $11608$ & $2.4972$ & exact (achieved)\\
$5$ & $6$ & $570$ & $\ge102736$ & $\le3.3514$ & lower bound\\
$6$ & $7$ & $5001$ & $\ge868456$ & $\le4.2719$ & lower bound\\
$7$ & $8$ & $41830$ & $\ge7141360$ & $\le5.2322$ & lower bound\\
\bottomrule
\end{tabular}
\label{tab:unsat}
\end{table}

\section{The unsaturated regime}
\label{sec:appunsat}

When $\binom m2>2^n-1$ no Sidon set exists, and the maximizer of $M_2$ is a near-Sidon set. For a $2$-dimensional subspace $L\le\mathbb{F}_2^n$ let
\begin{subequations}
 \begin{align}
 c_L:&=\#\bigl\{v\in\mathbb{F}_2^n \mid v+L\subseteq C\bigr\},\\
 N:&=\sum_{L}c_L,\\
 Q:&=\sum_{L}c_L^{2}.
 \end{align}
\end{subequations}
$N$ counts the affine $2$-planes contained in $C$ and $Q$ weights them by their $2$-subspace. The energies of Eq.~\eqref{eq:planes} follow by splitting into dependent and independent pairs: writing $p_x=A(x)/2$ for $x\ne0$, two disjoint $x$-pairs span a unique affine $2$-plane, each plane has three nonzero directions, and $\sum_{x\ne0}\binom{p_x}{2}=3N$. The independent pairs then contribute $96Q$ and the dependent pairs $3E(C)-2m^2$. Since $Q\ge N$ with equality when all affine $2$-planes of $C$ lie in pairwise distinct $2$-subspaces, and Sidon sets corresponds to those with $N=Q=0$, the maximization of $M_2$ reduces to minimizing $N$ and $Q$.

A Fourier lower bound makes this minimization explicit. With $\widehat C(z)=\sum_{w\in C}(-1)^{z\cdot w}$, we have
\begin{gather}
 N=\frac{1}{24}\Bigl(\frac{m^{4}}{2^n}-3m^{2}+2m\Bigr) +\frac{1}{24\cdot 2^n}\sum_{z\neq0}\bigl|\widehat C(z)\bigr|^{4},\\
 \sum_{z\neq0}\bigl|\widehat C(z)\bigr|^{2}=m(2^n-m),
\end{gather}
and, since $Q\ge N$ gives $E^{(1)}\ge7m^2-6m+168N$, Cauchy--Schwarz inequality yields
\begin{equation}
\begin{aligned}
 M_2 & \le4k-\log_2\Bigl\{7m^{2}-6m\\
 &\quad+7\Bigl[\frac{m^{4}}{2^n}-3m^{2}+2m +\frac{m^{2}(2^n-m)^{2}}{2^n-1}\Bigr]\Bigr\}.
\end{aligned}
\end{equation}

In the half-space regime $n=k+1$ (so $2^n=2m$ and $\sum_{z\ne0}|\widehat C(z)|^2=m^2=2^{2k}$), write $|\widehat C(z)|=2j$ with $r_j=\#\{z\ne0 \;\big|\; |\widehat C(z)|=2j\}$. The values $|\widehat C(z)|$ are even integers since $m$ is even. The minimal $N$ is then obtained by solving the following constrained integer minimization problem:
\begin{equation}
 \min \sum_j j^{4}r_j \quad \text{s.t.}\;
 \sum_j j^{2}r_j=2^{2k-2},\quad \sum_j r_j\le 2^{k+1}-1,
\end{equation}
subject to the integrality congruence
\begin{equation}
 \sum_j j^{4}r_j \equiv -\bigl(2^{3k-1}-3\cdot2^{2k}+2^{k+1}\bigr)2^{k-3} \pmod{3\cdot2^k},
\end{equation}
with $N=\frac{2^{3k-1}-3\cdot2^{2k}+2^{k+1}}{24}+\frac{1}{3\cdot2^k}\sum_j j^{4}r_j$. Table~\ref{tab:unsat} records the extremal data. For $k\ge5$ the values use $Q=N$.

Explicit maximizers are derived for the two exact rows. For $(n,k)=(4,3)$ the set
\begin{equation}
 C=\{0,e_1,e_2,e_1\oplus e_2,e_3,e_2\oplus e_3,e_4,e_1\oplus e_4\}\subseteq\mathbb{F}_2^4
\end{equation}
has difference profile $A=2$ (4 directions), $A=4$ (9), $A=6$ (2), and $N=Q=5$. Exhaustive search over all $\binom{16}{8}=12870$ subsets confirms $E^{(1)}=1240$ and $M_2=12-\log_2 1240\approx1.7239$ is maximal, with $10080$ maximizers sharing this profile. For $(n,k)=(5,4)$ the set (binary $5$-vectors)
\begin{equation}
 \begin{aligned}
 C=\{& 00010, 00011, 01010, 01011, 01100, 01111, 10001, 10011, \\
 &10100, 10101, 11000, 11001, 11010, 11011, 11110, 11111\}
 \end{aligned}
\end{equation}
has $N=Q=59$, $E^{(1)}=11608$, $M_2=16-\log_2 11608\approx2.4972$, with a single heavy direction $A=12$ forced by the congruence (the flat profile $\{6,8\}$ would violate integrality). Asymptotically, for $n=k+1$ and $k\to\infty$, $N\sim2^{3k-4}/3$, $E^{(1)}\sim7\cdot2^{3k-1}$, and
\begin{equation}
 M_2=k+1-\log_2 7+o(1),
\end{equation}
half the Sidon slope of $2$. The exact optimum for general $(n,k)$ beyond the half-space regime remains open.

\section{Nonstabilizerness capacity of Sidon code spaces}
\label{sec:appcapacity}

This appendix proves two lemmas used in Sec.~\ref{sec:codecapacity}. By Eq.~\eqref{eq:capmin} the capacity of a Sidon code is the minimum of the weighted energy $\mathcal E^{(1)}_{\psi}(C)$ over amplitudes and phases. Lemma~\ref{lem:sidonfourth} evaluates this energy in closed form for a Sidon set, and Lemma~\ref{lem:amplitude} solves the amplitude optimization.

\begin{lemma}[Fourth moment of a Sidon code space]
\label{lem:sidonfourth}
Let $C$ be a Sidon set and $|\psi\rangle=\sum_{w\in C}\psi_w|w\rangle$, with $p_w=|\psi_w|^2$, $\phi_w=\arg\psi_w$ and $s_q=\sum_w p_w^q$. Then
\begin{equation}
 \sum_{x,z}\bigl|\langle\psi|X^xZ^z|\psi\rangle\bigr|^4
 =2^n\Bigl[6\bigl(s_2^2-s_4\bigr)+\Bigl|\textstyle\sum_w p_w^2\mathrm{e}^{4\mathrm i\phi_w}\Bigr|^2\Bigr].
\label{eq:fourthfactor}
\end{equation}
\end{lemma}

\begin{proof}
For fixed $x$ write $c(x,z)=\langle\psi|X^xZ^z|\psi\rangle=\sum_w\overline{\psi_w}\psi_{w\oplus x}(-1)^{z\cdot w}=\hat{f}_x(z)$ with $f_x(w)=\overline{\psi_w}\psi_{w\oplus x}$ supported on $S_x=C\cap(C\oplus x)$, and use $\sum_z|\hat f(z)|^4=2^n\sum_u|\sum_w f(w)\overline{f(w\oplus u)}|^2$. The $x=0$ slice is $f_0=p$, with $u=0$ term $s_2^2$. Since every nonzero difference of a Sidon set is realized by a unique pair, the $u\ne0$ terms contribute $4\sum_{w<w'}p_w^2p_{w'}^2$, so the slice energy is $3s_2^2-2s_4$. For $x\ne0$ the slice is a single pair $\{w,w'\}$, with $f_x(w)=r\mathrm{e}^{\mathrm i\theta}$, $f_x(w')=r\mathrm{e}^{-\mathrm i\theta}$, $r^2=p_wp_{w'}$, $\theta=\phi_{w'}-\phi_w$, and energy $4r^4(1+\cos^2 2\theta)=2p_w^2p_{w'}^2(3+\cos4\theta)$. Summing over all $x$, with $\sum_{w<w'}p_w^2p_{w'}^2=\tfrac12(s_2^2-s_4)$ and $|\sum_w p_w^2\mathrm{e}^{4\mathrm i\phi_w}|^2=s_4+2\sum_{w<w'}p_w^2p_{w'}^2\cos4\theta_{ww'}$, gives Eq.~\eqref{eq:fourthfactor}.
\end{proof}

The phases are then optimized out in closed form. For fixed amplitudes, the quartic phases $4\phi_w$ are chosen to close the polygon of side lengths $p_w^2$, so that
\begin{equation}
 \min_{\phi}\Bigl|\textstyle\sum_w p_w^2\mathrm{e}^{4\mathrm i\phi_w}\Bigr|^2=\bigl[\max\bigl(0,\,2p_{\max}^2-s_2\bigr)\bigr]^2 .
\end{equation}
It remains to minimize (over amplitudes) the functional
\begin{equation}
 F(p)=6\bigl(s_2^2-s_4\bigr)+\bigl[\max\bigl(0,\,2p_{\max}^2-s_2\bigr)\bigr]^2 ,
\label{eq:F}
\end{equation}
on the simplex $\Delta_m=\{p\in\mathbb R_{\ge0}^m \mid \sum_w p_w=1\}$.

\begin{lemma}[Uniform amplitudes optimize the Sidon capacity]
\label{lem:amplitude}
For $m\ge4$,
\begin{equation}
 \min_{p\in\Delta_m}F(p)=\frac{6(m-1)}{m^3},
\end{equation}
attained uniquely at the uniform distribution $p_w\equiv1/m$. The exceptional values are $F=2/3$ at $m=2$ and $F=7/16$ at $m=3$.
\end{lemma}

\begin{proof}
The minimizer is interior, since a boundary point (some $p_w=0$) reduces the dimension to $m'<m$, where the minimum is at least $6(m'-1)/(m')^3\ge6(m-1)/m^3$ for $m'\ge4$ (and $2/3$ or $7/16$ at $m'=2,3$, both larger). Write $q=p_{\max}$. In the region $2q^2\le s_2$, where $F=6(s_2^2-s_4)$, the Lagrange equation is $p_w(s_2-p_w^2)=\mathrm{const.}$, and since $t\mapsto t(s_2-t^2)$ is strictly concave the coordinates take at most two values. In the region $2q^2>s_2$ the maximum $q$ is unique (two maxima would give $s_2\ge2q^2$), the remaining coordinates obey $p_w(7s_2-2q^2-6p_w^2)=\mathrm{const.}$, and $t\mapsto t(7s_2-2q^2-6t^2)$ is concave, so again at most two values. Two distinct small values $a>b$, occurring $j$ and $m-1-j$ times, would force $7s_2-2q^2=6(a^2+ab+b^2)$ and hence $5q^2+(7j-6)a^2+(7(m-1-j)-6)b^2=6ab$ with $q>a>b$, whose left side exceeds $6a^2>6ab$. Thus the minimizer is two-valued, $p=(a^{(r)},b^{(m-r)})$ with $ra+(m-r)b=1$ and $a\in[1/m,1/r]$. For $r\ge2$ one always has $2a^2\le s_2$ and
\begin{equation}
 \frac{\mathrm{d}F}{\mathrm{d}a}=24r(a-b)\bigl[(r-1)a^2+(m-r-1)b^2-ab\bigr]\ge0 ,
\end{equation}
so $F$ is increasing in $a$ and minimized at $a=1/m$, the uniform point. For $r=1$ (one large amplitude $a$, $m-1$ equal small amplitudes $b=(1-a)/(m-1)$) the boundary $2a^2=s_2$ is $a=q_0:=(1+\sqrt{m-1})^{-1}$. On $[1/m,q_0]$, $\mathrm{d}F/\mathrm{d}a=24(a-b)b[(m-2)b-a]\ge0$ since $a\le\sqrt{m-1}\,b\le(m-2)b$. On $[q_0,1]$, $\mathrm{d}F/\mathrm{d}a=4b^3G(a/b)$ with $G(\rho)=\rho^3-5\rho^2+5(m-1)\rho-(7m-13)$, strictly increasing on $\rho\ge\sqrt{m-1}$ and positive there since $G(\sqrt{m-1})=6(m-1)^{3/2}-12m+18\ge18\sqrt3-30>0$. Hence $F$ decreases to $a=1/m$ and then increases. Its transition value is $F(q_0)=6(3m-4)/[(m-1)(1+\sqrt{m-1})^4]\ge6(m-1)/m^3$, and the boundary $b=0$ gives $F=6(r-1)/r^3\ge6(m-1)/m^3$. The minimizer is therefore uniform.
\end{proof}

\end{document}